\documentclass[10pt,twocolumn,letterpaper]{article}
\usepackage{graphicx}
\usepackage{subfig} 
\usepackage{float}
\usepackage{multirow}
\usepackage{times}
\usepackage{epsfig}
\usepackage{graphicx}
\usepackage{algorithm}
\usepackage{amsthm}
\newtheorem{theorem}{Theorem}
\usepackage{amsmath}  
\usepackage{algorithmic}
\usepackage{amsmath}
\usepackage{amssymb}
\usepackage{booktabs}
\usepackage{threeparttable}
\usepackage{bm}
\usepackage{marvosym}
\usepackage[dvipsnames]{xcolor}
\usepackage[accepted]{cvpr}      
\definecolor{cvprblue}{rgb}{0.21,0.49,0.74}
\usepackage[pagebackref,breaklinks,colorlinks,allcolors=cvprblue]{hyperref}

\def\paperID{2164} 
\def\confName{CVPR}
\def\confYear{2025}

\title{ECVC: Exploiting Non-Local Correlations in Multiple Frames for\\Contextual Video Compression}

\author{    Wei Jiang, Junru Li,
Kai Zhang, Li Zhang\textsuperscript{\Letter}\\
Bytedance\\
{\tt\small \{jiangwei.lvc, lijunru, zhangkai.video, lizhang.idm\}@bytedance.com}
}

\begin{document}
\maketitle
\begin{abstract}
    In Learned Video Compression (LVC), improving inter prediction, such as enhancing temporal context mining and mitigating accumulated errors, is crucial for boosting rate-distortion performance. 
    Existing LVCs mainly focus on mining the temporal movements while
    neglecting non-local correlations among frames. Additionally, 
    current contextual video compression models use a single reference frame, which is insufficient for handling complex movements.
    To address these issues, we propose leveraging non-local correlations across multiple
    frames to enhance temporal priors, significantly boosting rate-distortion performance.
    To mitigate error accumulation, we introduce a partial cascaded fine-tuning strategy that supports fine-tuning on full-length sequences with constrained computational resources. This method reduces the train-test mismatch in sequence lengths and significantly decreases accumulated errors.
    Based on the proposed techniques, we present a video compression scheme ECVC.
  Experiments demonstrate that our ECVC achieves state-of-the-art performance,
  reducing $10.5\%$ and $11.5\%$ more bit-rates than previous SOTA method DCVC-FM
  over VTM-13.2 low delay B (LDB) under the intra period (IP) of $32$ and $-1$\footnote[1]{Only one intra frame when compressing a sequence.}, respectively.
  \end{abstract}    
\section{Introduction}
Video coding aims to compactly represent the visual signals while maintaining acceptable reconstructed quality.
Traditional video coding, such as H.266/VVC~\cite{bross2021overview}, is established on the block-based hybrid coding framework, which has been developed for several decades.
\begin{figure}[t]
  \centering
  \subfloat{
    \includegraphics[scale=0.14]{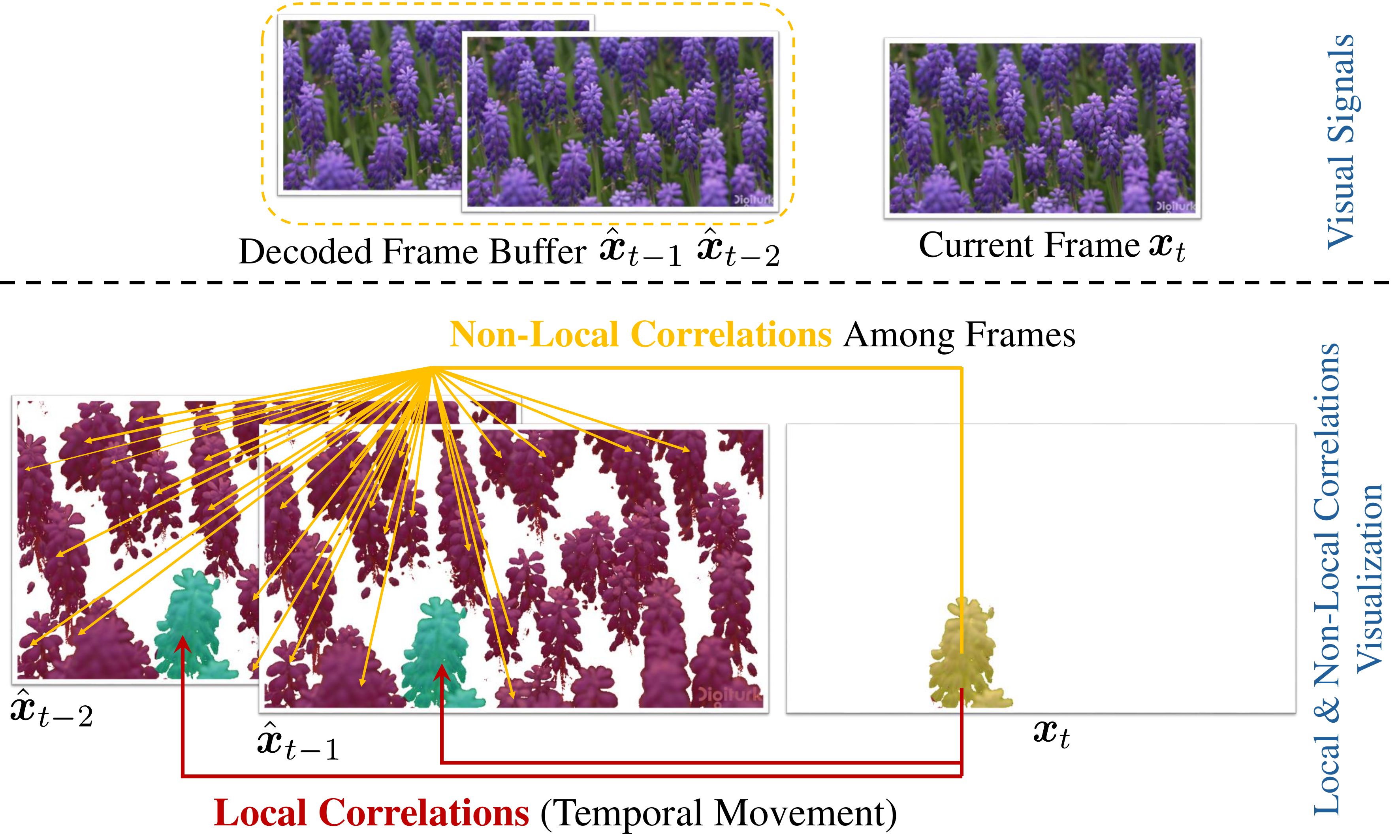}}
  \caption{Visualization of local (red lines) correlations $\hat{\boldsymbol{C}}^{\ell}$ and non-local (orange lines) correlations $\hat{\boldsymbol{C}}^{n\ell}$
  in sequence ``HoneyBee'' in UVG dataset~\cite{mercat2020uvg}.
  $\hat{\boldsymbol{C}}^{n\ell}$ is defined as the distant correlations without explicit movement between the current frame and reference frames.
  $\hat{\boldsymbol{C}}^{n\ell}$ is beneficial for temporal context mining due to lower conditional entropy for $\boldsymbol{x}_t$,
  formulated as $\mathcal{H}(\boldsymbol{x}_t|\hat{\boldsymbol{C}}^{\ell},\hat{\boldsymbol{C}}^{n\ell})\leq \mathcal{H}(\boldsymbol{x}_t|\hat{\boldsymbol{C}}^{\ell})$, where $\mathcal{H}$ is the Shannon entropy.}
  \label{fig:ls}
\end{figure}
However, as the performance improvement of traditional codecs approaches a plateau,
learned video compression (LVC) has emerged as a promising alternative, attracting significant attention from both academic and industrial fields.
The end-to-end optimization~\cite{balle2016end,balle2018variational,mentzer2022vct,li2021deep,lu2019dvc,lu2024deep,yang2023introduction,ma2019image} of LVCs offers 
the potential to surpass handcrafted traditional codecs with redundancy removal.\par
Early LVCs~\cite{lu2019dvc,hu2021fvc, guo2023learning,liu2020learned, rippel2021elf, yang2020learning,yang2024adaptive,pourreza2021extending,pourreza2023boosting,lin2022dmvc,yang2023insights,lu2020end} follow the residual coding paradigm, 
wherein the residuals between the predicted frame (via estimated motion) and the current frame, along with motion information, are compressed.
Recently, the conditional coding paradigm~\cite{li2021deep} has emerged,
which leverages temporal information as a prior, enabling a reduction in conditional entropy and offering more flexibility compared to the predefined subtraction used in residual coding.
Subsequent research has explored various techniques~\cite{li2022hybrid,ho2022canf,sheng2022temporal,xiang2023mimt,Qi_2023_CVPR,zhai2024hybrid,yang2024ucvc,jiang2024lvc,sheng2024spatial,alexandre2023hierarchical,wang2024exp,chen2024maskcrt,chen2024bcanf, li2023neural, li2024neural,guo2023enhanced}  to enhance the performance of conditional coding.
Notably, recent DCVC-DC~\cite{li2023neural} has outperformed the H.266/VVC low delay B (LDB) coding with an intra period (IP) 
of $32$ by exploiting diverse temporal contexts. 
Additionally, DCVC-FM~\cite{li2024neural} surpasses the H.266/VVC LDB configuration~\cite{bross2021overview} 
in long coding chains, particularly when the IP is $-1$, by employing temporal feature modulation.\par
Inter prediction is a long-standing challenge in video coding, aimed at characterizing temporal movements in videos to eliminate temporal redundancies. This temporal information serves as a beneficial contextual prior for current frame coding. 
Recent advancements in LVC, such as DCVC-DC~\cite{li2023neural} and DCVC-FM~\cite{li2024neural}, have utilized offset diversity~\cite{chan2021understanding} for inter prediction, enabling precise capture of local small-scale movements between frames.
However, existing LVC approaches generally overlook non-local correlations.
For instance, as illustrated in Figure~\ref{fig:ls}, while the movement of one flower can be accurately estimated,
the similarities among different flowers, representing non-local correlations,
remain challenging to capture using optical flow~\cite{ranjan2017optical} / DCN~\cite{dai2017deformable}. 
Moreover, multiple reference frames, a technique commonly adopted in traditional video coding paradigms~\cite{huang2006analysis}, are not incorporated in DCVC-DC or DCVC-FM. Exploiting non-local correlations across multiple frames could potentially enhance model performance.
\par 
To address the limitations on non-local correlations across multiple frames, we propose the \textbf{M}ultiple Frame \textbf{N}on-\textbf{L}ocal \textbf{C}ontext Mining (MNLC) to capture
the local and non-local contexts for the $t$-th frame $\boldsymbol{x}_t$
from two reference frames $\hat{\boldsymbol{x}}_{t-1}, \hat{\boldsymbol{x}}_{t-2}$,
achieving enhanced rate-distortion performance with moderate complexity increases.
Specifically,
regarding the multiple frames' local correlations, the offset diversity~\cite{chan2021understanding}
is employed to capture local correlations in $\hat{\boldsymbol{x}}_{t-1}$. 
The former motion between $\hat{\boldsymbol{x}}_{t-2}$ and $\boldsymbol{x}_{t-1}$ is
reused to exploit local correlations in
$\hat{\boldsymbol{x}}_{t-2}$.
To capture non-local correlations, we propose the \textbf{M}ulti-\textbf{H}ead \textbf{L}inear
\textbf{C}ross \textbf{A}ttention (MHLCA).
Leveraging the flexibility of conditional coding, 
our model learns non-local correlations through cross-attention 
between the current mid-feature during transform and multiple temporal priors. 
The attention mechanism computes similarity among all elements, 
capturing correlations between distant elements. 
Additionally, we employ the linear decomposition of vanilla attention~\cite{vaswani2017attention} to mitigate high complexity.
By incorporating both local and non-local priors from multiple frames,
our approach significantly enhances the performance.\par 
Error propagation is one of the key issues causing quality degradation in inter frame coding, 
particularly in scenarios involving long prediction chains (\textit{e.g.}, video conferencing, monitoring scene).
The coding distortions in the previously coded frames are sequentially propagated and accumulated to the current frame through inter prediction, significantly damaging the coding efficiency of the current frame.
To reduce accumulated prediction errors, DCVC-FM employs propagated feature refreshment, 
significantly outperforming DCVC-DC. 
However, the temporal context refreshment in DCVC-FM is not mature enough in achieving consistent performance improvements, especially on videos with fast movements. 
The primary factor influencing error accumulation is the train-test mismatch in sequence lengths. 
Most existing LVCs are trained with only $6$ to $7$ frames due to the limited computational resources,
whereas the testing sequence length may reach to hundreds of frames.
Although DCVC-FM attempts to train on long sequences, the vanilla training strategy is still a heavy demand for training resources.
To address this issue, we 
propose the \textbf{P}artial \textbf{C}ascaded \textbf{F}inetuning Strategy (PCFS),
enabling finetuning of LVC on full-length sequences, 
thereby significantly reducing accumulated errors within the computational resources budget.
\par
Based on the proposed techniques, we introduce \textbf{E}xploiting Non-Local \textbf{C}orrelations
in multiple frames for 
Contextual \textbf{V}ideo \textbf{C}ompression (ECVC), 
which is established beyond the DCVC-DC. 
Experiments demonstrate that our ECVC achieves state-of-the-art performance,
reducing $7.3\%$ and $10.5\%$ more bit-rates than DCVC-DC and DCVC-FM
over VTM-13.2 LDB, respectively under IP $32$.
Additionally, ECVC reduces $11.5\%$ more bit-rate than DCVC-FM over VTM-13.2 LDB when the IP is $-1$.
The contributions of this paper are summarized as follows:
\begin{itemize}
  \item \textbf{Enhanced Temporal Priors:} We analyze the potential of exploiting spatial \textit{non-local correlations} in multiple frames and propose a novel multiple-frame 
  non-local context mining module. This module enables the model to aggregate more
  temporal priors to boost performance. To our knowledge, we are the \textit{first}
  in the learned video compression community to exploit \textit{non-local correlations}
  in multiple frames.
  \item \textbf{Mitigated Accumulated Errors:} We address the train-test mismatch in sequence length 
  and propose the novel partial cascaded finetuning strategy, 
  enabling finetuning on unlimited-length sequences with error awareness.
  \item \textbf{Experimental Validation:} The ECVC achieves state-of-the-art performance 
  under IP $32$ and IP $-1$ settings. Specifically, 
  our ECVC reduces $10.5\%$ more bit-rates than DCVC-FM~\cite{li2024neural} over VTM-13.2 LDB under IP $32$.
Additionally, ECVC reduces $11.5\%$ more bit-rates than DCVC-FM~\cite{li2024neural} over VTM-13.2 LDB when the IP is $-1$.
\end{itemize}
\section{Related Works}
DVC~\cite{lu2019dvc} is one of the pioneer learned video compression frameworks, wherein an optical flow net~\cite{ranjan2017optical} is employed for motion estimation. Then the residuals between predicted frame and original frame are calculated and compressed along with the motion information by neural networks.
To further enhance the rate-distortion performance, advanced techniques such 
scale-space warping~\cite{agustsson2020scale}, adaptive flow coding~\cite{hu2020improving},
deformable convolutions (DCN)~\cite{dai2017deformable,hu2021fvc}, coarse-to-fine mode prediction~\cite{hu2022coarse} and
pixel-to-feature motion compensation~\cite{shi2022alphavc} are investigated, leading to the improvement of the prediction accuracy.\par
Recently, Li~\textit{et al.}~\cite{li2021deep} proposed the conditional coding framework DCVC, wherein
the temporal information serves as the context for current frame coding instead of explicitly residual coding.
In this way, the network learns the correlations between temporal context and current frame automatically.
The conditional coding is more flexible and thus 
breaks the performance bound of residual coding.
DCVC-TCM~\cite{sheng2022temporal} employs the multi-scale temporal contexts 
and temporal propagation mechanism to exploit more temporal priors.
Li~\textit{et al.}~\cite{li2022hybrid} further enhances the performance via advanced dual spatial contexts
for entropy modeling.
In DCVC-DC~\cite{li2023neural}, a hierarchical quality structure is employed to
alleviate the error propagation.
In DCVC-FM~\cite{li2024neural}, the temporal propagated contexts are periodically refreshed
to further enhance the performance with a long prediction chain
(\textit{e.g.} one intra frame setting).\par
However, existing DCVC series only consider the local correlations in inter prediction. The non-adjacent correlations are ignored, 
which limits the potential for performance improvement of LVCs.
In addition, considering that videos contain complex scenes and motions, such as fast movements~\cite{fan2016motion}, affine motions~\cite{zhang2018improved}, and occlusion~\cite{huang2006analysis}, one single reference frame may not effectively capture such complicated scenario. Multiple reference frames can provide richer motion information.
Moreover, although DCVC-FM achieves remarkable performance under a long prediction chain, 
the temporal context refreshment in DCVC-FM is not mature enough to achieve consistent performance improvements, especially on videos with fast movements. 
There is room for further 
reducing the error propagation. 
\section{Method}
\subsection{Overview of ECVC}
The proposed ECVC builds upon the DCVC-DC~\cite{li2023neural}, 
but focuses more on exploiting non-local correlations in multiple frames. 
The architecture of ECVC is presented in Figure~\ref{fig:arch:elvc} and Figure~\ref{fig:MNLC}.
The method processes the $t$-th frame $\boldsymbol{x}_t$ by firstly converting the propagated features $\hat{\boldsymbol{F}}_{t},\hat{\boldsymbol{F}}_{t-1}$ into multi-scale features $\hat{\boldsymbol{F}}_t^i,\hat{\boldsymbol{F}}_{t-1}^i \in \mathbb{R}^{ \frac{H}{2^i} \times \frac{W}{2^i}\times d^i_f},0\leq i \leq 2,$ targeting at the coding of mid-feature $\boldsymbol{y}^i_t \in \mathbb{R}^{ \frac{H}{2^i}\times \frac{W}{2^i}\times d^i_e}$, where $H, W$ are the height and width. $d_f^i, d_e^i$ denote channel numbers.
A key innovation is the Multiple Frame Non-Local Context Mining, which leverages multiple reference frames to extract local and non-local contexts. 
For local contexts $\hat{\boldsymbol{C}}_{t-1\rightarrow t}^{\ell,i}$ and $\hat{\boldsymbol{C}}_{t-2\rightarrow t}^{\ell,i}$, we use motion vectors $\hat{\boldsymbol{v}}_{t}$ and $\hat{\boldsymbol{v}}_{t-1}$ to capture local contexts from the multi-scale features $\hat{\boldsymbol{F}}_t^i$ and $\hat{\boldsymbol{F}}_{t-1}^i$ via the offset diversity and proposed multi-scale refinement module, yielding priors from \textit{two} reference frames. 
For non-local contexts, the Multi-Head Linear Cross
Attention is proposed to capture non-local contexts $\hat{\boldsymbol{C}}_{t-1\rightarrow t}^{n\ell,i}$ and $\hat{\boldsymbol{C}}_{t-2\rightarrow t}^{n\ell,i}$ among $\boldsymbol{y}^i_t, \hat{\boldsymbol{F}}_t^i, \hat{\boldsymbol{F}}_{t-1}^i$.
Those local and non-local contexts are then used for conditional coding of $\boldsymbol{y}_t^i$. The decoding process mirrors the encoding process but uses $\hat{\boldsymbol{y}}_t^i$ as input instead of $\boldsymbol{y}_t^i$. Besides, a Partial Cascaded Finetuning Strategy (PCFS) is proposed to further mitigate error accumulation, as depicted in Figure~\ref{fig:pcfs}.
\begin{figure}[t]
  \centering
  \subfloat{
    \includegraphics[scale=0.31]{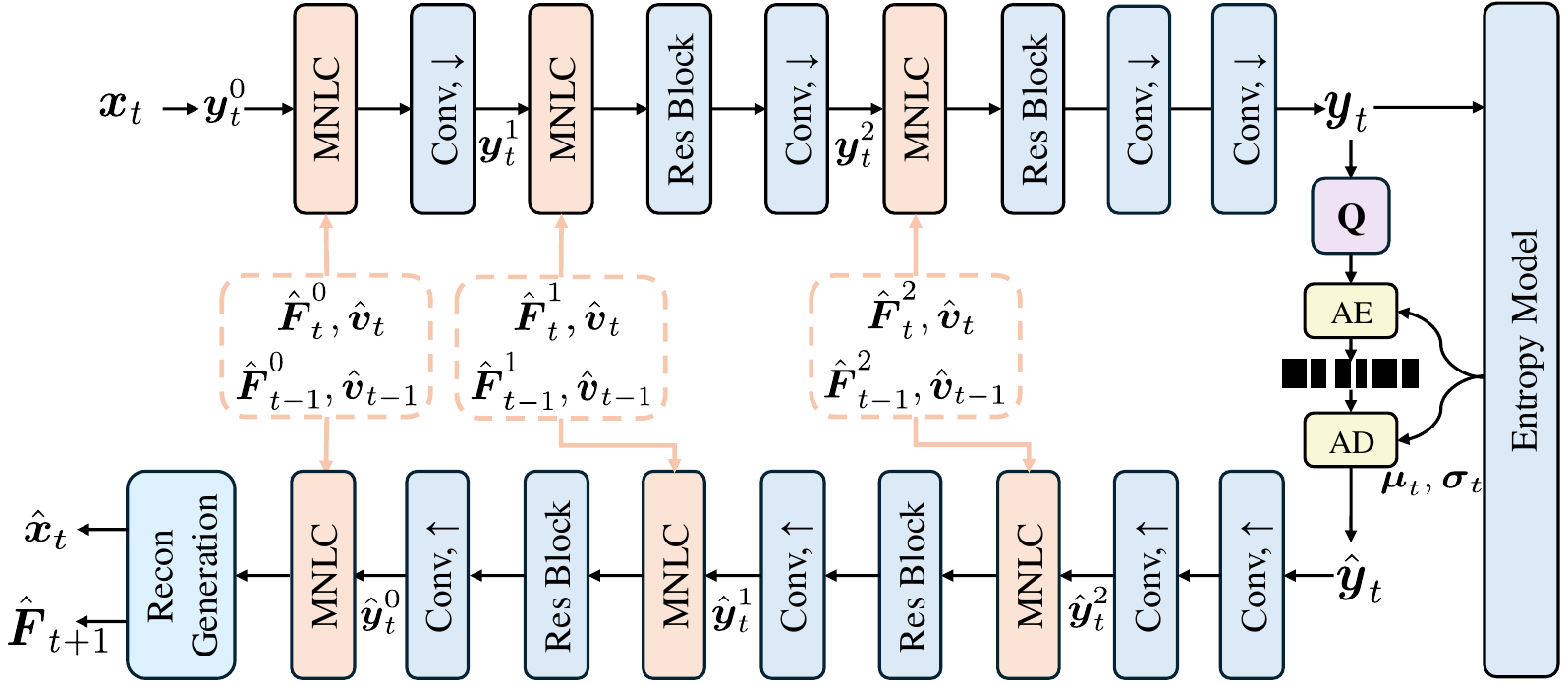}}
  \caption{Illustration of contextual encoder and decoder.
  $\boldsymbol{x}, \hat{\boldsymbol{x}}$ are input frame and reconstructed frame.
  $\hat{\boldsymbol{F}}_{t+1}$ is the propagated feature for coding $\boldsymbol{x}_{t+1}$.
  $\hat{\boldsymbol{v}}_t, \hat{\boldsymbol{v}}_{t-1}$ are motion vectors.
  $\{\boldsymbol{y}_t^0, \boldsymbol{y}_t^1, \boldsymbol{y}_t^2\}, \{\hat{\boldsymbol{y}}_t^0, \hat{\boldsymbol{y}}_t^1, \hat{\boldsymbol{y}}_t^2\}$ 
  are mid-features during encoding and decoding.
  ``Res Block'' and ``Recon Generation'' are adopted from DCVC-DC. ``Q" is quantization. ``AE" denotes arithmetic encoding and ``AD" denotes arithmetic decoding. $\boldsymbol\mu_t$ and $\boldsymbol{\sigma}_t$ 
  are estimated means and scales of $\hat{\boldsymbol{y}}_t$ by the entropy model for AE/AD.}
  \label{fig:arch:elvc}
\end{figure}
\subsection{Exploiting Non-Local Correlations}
As illustrated in Figure~\ref{fig:ls},
there are both local and non-local correlations between the current frame and reference frames.
Existing methods mainly focus on capturing local correlations (\textit{i.e.} temporal movement).
It is desirable to capture both local and
non-local correlations (\textit{i.e.} similar regions without explicit movement) to boost the performance.
In addition, considering that videos contain complex scenes and motions~\cite{huang2006analysis,zhang2018improved}, using a single reference frame may not effectively capture complicated scenarios. Multiple reference frames can provide richer motion information.
To tackle these issues, we propose the Multiple Frame Non-Local Context Mining,
where two reference frames are employed on account of maintaining the complexity.
The process of conditionally coding $\boldsymbol{y}_t^i$ is employed as an example and depicted in Figure~\ref{fig:MNLC}.
Specifically,
the offset diversity 
is adopted and extended for two reference frames. 
$\hat{\boldsymbol{F}}_{t}^i$ is first warped by decoded motion vector $\hat{\boldsymbol{v}}_t$
to $\overline{\boldsymbol{F}}_{t}^i$
and then refined by offset diversity to $\hat{\boldsymbol{C}}_{t-1\rightarrow t}^{\ell,i}$. The process is:
\begin{equation}
  \begin{aligned}
    \hat{\boldsymbol{C}}_{t-1\rightarrow t}^{\ell,i} &= \textrm{OffsetDiversity}(\textrm{Warp}(\hat{\boldsymbol{F}}_{t}^i, \hat{\boldsymbol{v}}_t),\hat{\boldsymbol{v}}_t).
  \end{aligned}
\end{equation}
To capture additional local contexts from $\hat{\boldsymbol{F}}_{t-1}^i$,
the previous local context $\hat{\boldsymbol{C}}_{t-2\rightarrow t-1}^{\ell,i}$ is reused to avoid \textit{additional motion bits}.
$\hat{\boldsymbol{C}}_{t-2\rightarrow t-1}^{\ell, i}$ is firstly warped by $\hat{\boldsymbol{v}}_t$
and then refined to local context $\hat{\boldsymbol{C}}_{t-2\rightarrow t}^{\ell, i}$ by a multi-scale refinement module
which can be formulated as:
\begin{equation}
  \begin{aligned}
    \hat{\boldsymbol{C}}_{t-2\rightarrow t}^{\ell, i} &= \textrm{MultiScaleRefine}(\textrm{Warp}(\hat{\boldsymbol{C}}_{t-2\rightarrow t-1}^{\ell, i},\hat{\boldsymbol{v}}_t)).
  \end{aligned}
\end{equation}\par
To capture non-local correlations, it is necessary to compare \textit{all} elements,
which implies the receptive field should be large enough. 
The offset diversity is not sufficient due to its limited kernel sizes of convolutions.
The necessity of large receptive fields inspires us to employ attention mechanisms.
Thanks to the flexible conditional coding paradigm, the network itself is able to learn the non-local correlations by
cross attention between the current feature and given temporal contexts.
For simplicity, we omit the Depth-wise Res Block in Figure~\ref{fig:MNLC}.
The non-local contexts $\hat{\boldsymbol{C}}_{t-1\rightarrow t}^{n\ell, i}, \hat{\boldsymbol{C}}_{t-2\rightarrow t}^{n\ell, i}$
from $\hat{\boldsymbol{F}}^i_t,\hat{\boldsymbol{F}}^i_{t-1}$ for $\boldsymbol{y}_t^i$
can be captured by, 
\begin{figure}[t]
  \centering
  \subfloat{
  \includegraphics[scale=0.3]{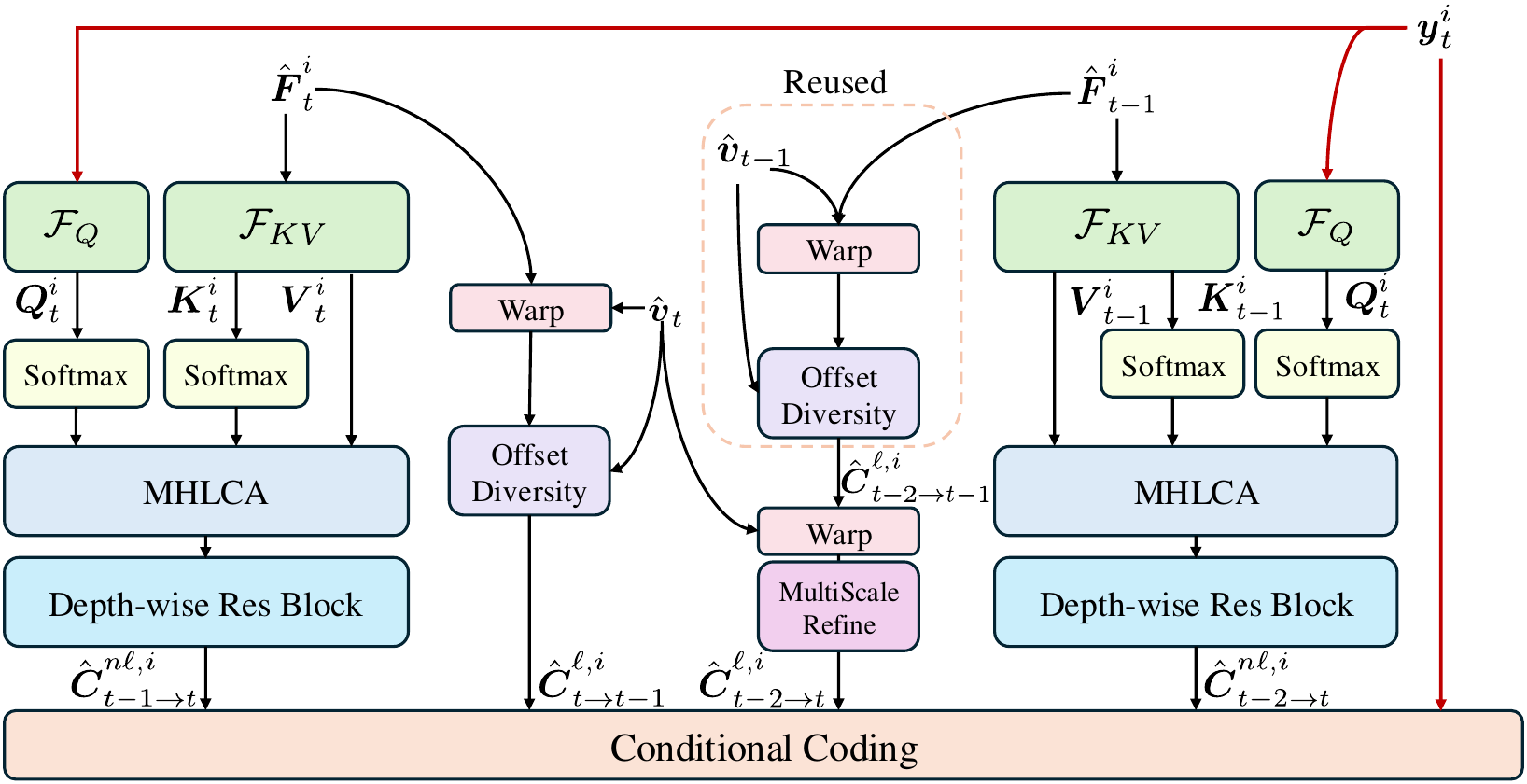}}
  \caption{Proposed {M}ultiple Frame {N}on-{L}ocal {C}ontext Mining (MNLC) for conditional coding of the $\boldsymbol{y}_t^i$.}
  \label{fig:MNLC}
\end{figure}
\begin{equation} 
  \begin{aligned}
    \hat{\boldsymbol{C}}_{t-1\rightarrow t}^{n\ell, i} &= \underbrace{\underbrace{\textrm{Softmax}(\boldsymbol{Q}_{t}^i(\boldsymbol{K}_{t}^i)^{\top})}_{\mathcal{O}(H^2W^2d)}\boldsymbol{V}_{t}^i}_{\mathcal{O}(H^2W^2d + H^2W^2d)},\\
    \hat{\boldsymbol{C}}_{t-2\rightarrow t}^{n\ell, i} &= \textrm{Softmax}(\boldsymbol{Q}_{t}^i(\boldsymbol{K}_{t-1}^i)^{\top})\boldsymbol{V}_{t-1}^i,
  \end{aligned}
  \label{eq:quad_lr}
\end{equation}
where $\boldsymbol{Q}_{t}^i=\mathcal{F}_Q(\boldsymbol{y}_t^i)$, 
$\boldsymbol{K}_{t}^i,\boldsymbol{V}_{t}^i=\mathcal{F}_{KV}(\hat{\boldsymbol{F}}^i_t)$,
$\boldsymbol{K}_{t-1}^i,\boldsymbol{V}_{t-1}^i=\mathcal{F}_{KV}(\hat{\boldsymbol{F}}^i_{t-1})$, $\boldsymbol{Q}_{t}^i,\boldsymbol{K}_{t}^i,\boldsymbol{V}_{t}^i,\boldsymbol{K}_{t-1}^i,\boldsymbol{V}_{t-1}^i \in \mathbb{R}^{\frac{H}{2^i}\times\frac{W}{2^i}\times d}$, 
$d$ is the output channel number of embedding.
$\mathcal{F}_Q, \mathcal{F}_{KV}$ are the embedding layer for queries, keys, and values.
However, the resolutions of input sequences could be $2$K or $4$K, which means the complexity could not be too high.
The complexity of Equation~(\ref{eq:quad_lr}) is $\mathcal{O}(H^2W^2d)$.
To make the exploiting of non-local correlations possible, it is important to reduce the complexity.
The key of Equation~(\ref{eq:quad_lr}) is the \textit{non-negativity} of the 
attention map.
Therefore the generalized attention mechanism could be 
$\textrm{Similarity}(\boldsymbol{Q}_t^i,(\boldsymbol{K}_t^i)^\top)\boldsymbol{V}$.
Inspired by recent advancements in 
linear attention~\cite{katharopoulos2020transformers,shen2021efficient,han2023flatten,jiang2023mlicpp,jiang2022mlic}, 
it is promising to apply a non-negative projection functions $\psi, \phi$ to
$\boldsymbol{Q}_t^i$ and $\boldsymbol{K}_t^i$ to make $\psi(\boldsymbol{Q}_t^i)\phi(\boldsymbol{K}_t^i) \geq 0$.
The other key of Equation~(\ref{eq:quad_lr}) is the normalized value, which makes
it like the probabilities.
Inspired by the two factors, we employ the \textbf{M}ulti-\textbf{H}ead \textbf{L}inear \textbf{C}ross \textbf{A}ttention (MHLCA), 
which applies two independent softmax operations on row and column~\cite{shen2021efficient}
on queries and keys according to Theorem~\ref{theorem:efficient}.
\begin{figure*}[t]
  \centering
  \subfloat{
    \includegraphics[scale=0.1]{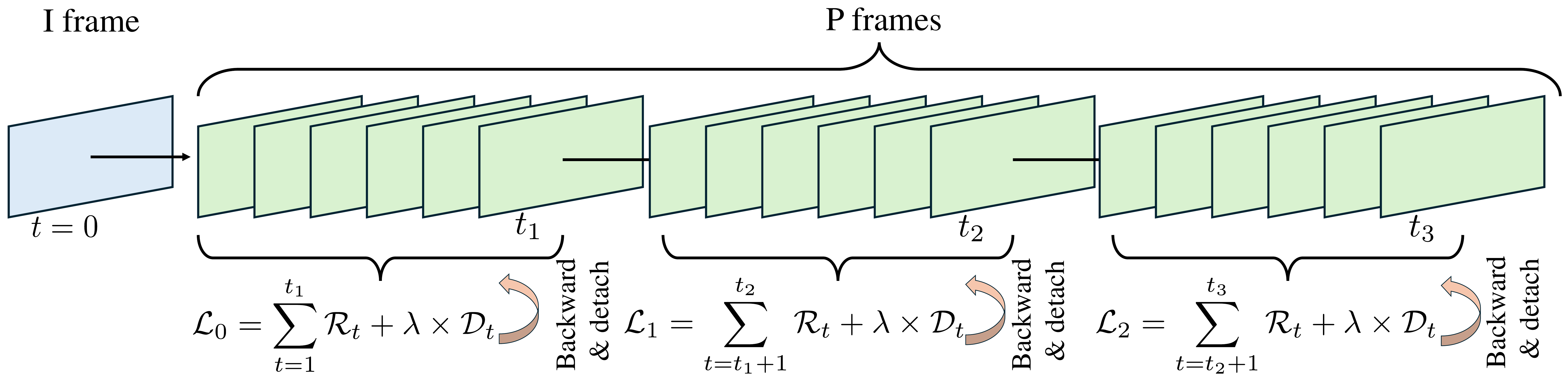}}
  \caption{Proposed Partial Cascaded Finetuning Strategy (PCFS). The I frame model is frozen during finetuning.}
  \label{fig:pcfs}
\end{figure*}
\begin{theorem}
  Like standard vanilla attention, each row of the implicit similarity matrix
  $\textrm{Softmax}_2(\boldsymbol{Q})\textrm{Softmax}_1(\boldsymbol{K})^\top$ 
  sums to 1, representing a normalized attention distribution across all positions.
  \label{theorem:efficient}
\end{theorem}
\begin{proof}
  Evidently, each row of the similarity matrix in standard attention sums to $1$.
  Let $\textrm{Softmax}_2(\boldsymbol{Q})=\boldsymbol{\mathcal{Q}}\in \mathbb{R}^{L\times C}$, 
  $\textrm{Softmax}_1(\boldsymbol{K})^\top=\boldsymbol{\mathcal{K}}^{\top}\in \mathbb{R}^{C\times L}$, where
  \begin{equation}
    \boldsymbol{\mathcal{Q}}=\begin{bmatrix}
  q_{1,1} & \cdots   & q_{1,C}   \\
  \vdots   & \ddots   & \vdots  \\
  q_{L,1} & \cdots\  & q_{L,C}  \\
       \end{bmatrix},
    \boldsymbol{\mathcal{K}}^\top=\begin{bmatrix}
      k_{1,1} & \cdots   & k_{1,L}   \\
      \vdots  & \ddots   & \vdots  \\
      k_{C,1} & \cdots\  & k_{C,L}  \\
           \end{bmatrix}.
  \end{equation}
  Due to the properties of the softmax operation, we have 
  $\sum_{i=1}^C q_{j, i}=1$, $1\leq j\leq L$; $\sum_{i=1}^L k_{j,i}=1$, 
  $1\leq j\leq C$.
  \begin{equation}
    \begin{aligned}
      \boldsymbol{\mathcal{Q}}\boldsymbol{\mathcal{K}}^\top
      =\begin{bmatrix}
        \sum_{i=1}^C q_{1,i}k_{i,1}  & \cdots   & \sum_{i=1}^C q_{1,i}k_{i,L}   \\
        \vdots & \ddots   & \vdots  \\
        \sum_{i=1}^C q_{L,i}k_{i,1} & \cdots\  & \sum_{i=1}^C q_{L,i}k_{i,L}  \\
             \end{bmatrix}.
    \end{aligned}
  \end{equation}
  Consider the sum of the $\ell$-th row ($1\leq \ell \leq L$):
  \begin{equation}
  \begin{aligned}
      \textrm{Sum}_{\ell} &= \sum_{i=1}^C q_{\ell,i}k_{i,1} + \sum_{i=1}^C q_{\ell,i}k_{i,2} + \cdots + \sum_{i=1}^C q_{\ell,i}k_{i,L}\\
      &= \left(q_{\ell, 1}k_{1,1} + q_{\ell, 2}k_{2,1} + \cdots + q_{\ell, C}k_{C,1}\right) \\
      &+ \cdots \\
      &+ \left(q_{\ell, 1}k_{1,L} + q_{\ell, 2}k_{2, L} + \cdots + q_{\ell, C}k_{C,L}\right)\\
      &= q_{\ell, 1}\underbrace{\sum_{i=1}^Lk_{1,i}}_{1} +  \cdots + q_{\ell, C}\underbrace{\sum_{i=1}^Lk_{C,i}}_{1}
      = \sum_{i=1}^C q_{\ell, i} =1.
  \end{aligned}
  \end{equation}
  Thus, each row of the similarity matrix sums to $1$, completing the proof.
\end{proof}
The softmax-based non-linear projection of queries and keys makes it able to 
compute the product of keys and values first, resulting in linear complexity $\mathcal{O}(HWd^2)$.
The non-local contexts $\hat{\boldsymbol{C}}_{t-1\rightarrow t}^{n\ell, i}, \hat{\boldsymbol{C}}_{t-2\rightarrow t}^{n\ell, i}$
can be learned \textit{linearly} by 
\begin{equation}
  \begin{aligned}
    \hat{\boldsymbol{C}}_{t-1\rightarrow t}^{n\ell, i} &= \underbrace{\textrm{Softmax}_2(\boldsymbol{Q}_{t}^i)\underbrace{\left(\textrm{Softmax}_1(\boldsymbol{K}_{t}^i)^\top\boldsymbol{V}_{t}^i\right)}_{\mathcal{O}(HWd^2)}}_{\mathcal{O}(HWd^2 + HWd^2)},\\
    \hat{\boldsymbol{C}}_{t-2\rightarrow t}^{n\ell, i} &= \textrm{Softmax}_2(\boldsymbol{Q}_{t}^i)\left(\textrm{Softmax}_1(\boldsymbol{K}_{t-1}^i)^\top\boldsymbol{V}_{t-1}^i\right).
  \end{aligned}
  \label{eq:quad_lr2}
\end{equation}
The captured local contexts $\hat{\boldsymbol{C}}_{t-1\rightarrow t}^{\ell,i}, \hat{\boldsymbol{C}}_{t-2\rightarrow t}^{\ell, i}$,
and non-local contexts $\hat{\boldsymbol{C}}_{t-1\rightarrow t}^{n\ell,i}, \hat{\boldsymbol{C}}_{t-2\rightarrow t}^{n\ell,i}$
are employed as priors to conditional coding $\boldsymbol{y}_t^i$.
The complexity of Equation~(\ref{eq:quad_lr2}) is $\frac{d}{HW}$ of that of Equation~(\ref{eq:quad_lr}).
If $H=1920, W=1080, d=48$, \textit{the complexity of Equation~(\ref{eq:quad_lr2}) is only $0.002\%$ of that of Equation~(\ref{eq:quad_lr})}.
The process of decoding is similar to that of encoding, except that the input $\boldsymbol{y}_t^i$ is replaced with $\hat{\boldsymbol{y}}_t^i$.
\subsection{Enhanced Long Coding Chain Adaptation}
In low delay scenarios (\textit{e.g.}, video conferencing, monitoring scene), one sequence may contain hundreds of frames with only one 
initial intra frame under long coding chains.
For enhanced long coding chain adaptation, it is desirable to have the model experience long coding chains during training.
However, most existing LVCs~\cite{lu2019dvc,sheng2022temporal,li2022hybrid,liu2023mmvc} are trained using a maximum of $6-7$ frames,
resulting in a significant train-test mismatch in terms of intra period.
This discrepancy leads to the accumulation of substantial errors.
Addressing the challenge of reducing accumulated errors within limited computational resources remains a critical area of investigation.
Conventionally, LVC models are trained using a cascaded loss function~\cite{lu2020content,sheng2022temporal}:
\begin{equation}
    \mathcal{L}=\sum_{t=1}^T \mathcal{R}_t + \lambda\times \mathcal{D}_t,
\end{equation}
where $T$ is 
the length of frames used for training, $\mathcal{R}_t$ is the frame bit-rate, $\mathcal{D}_t$ is the frame distortion.
\begin{figure}
  \centering
  \subfloat{
    \includegraphics[scale=0.29]{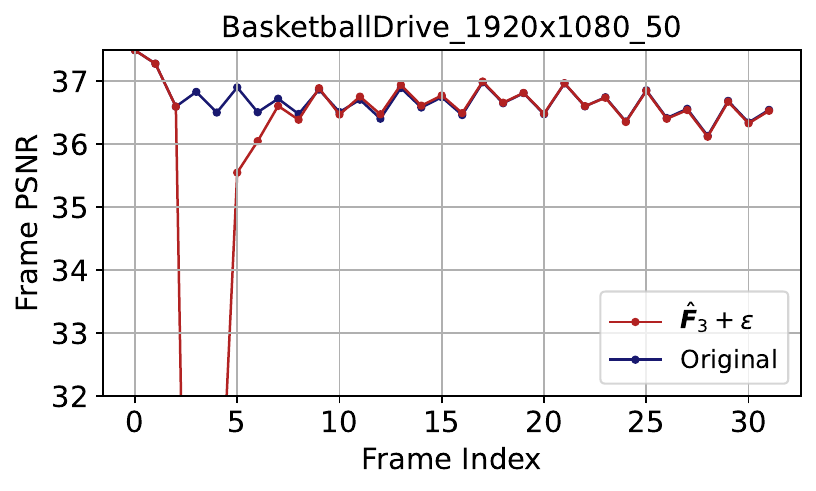}}
  \subfloat{
    \includegraphics[scale=0.29]{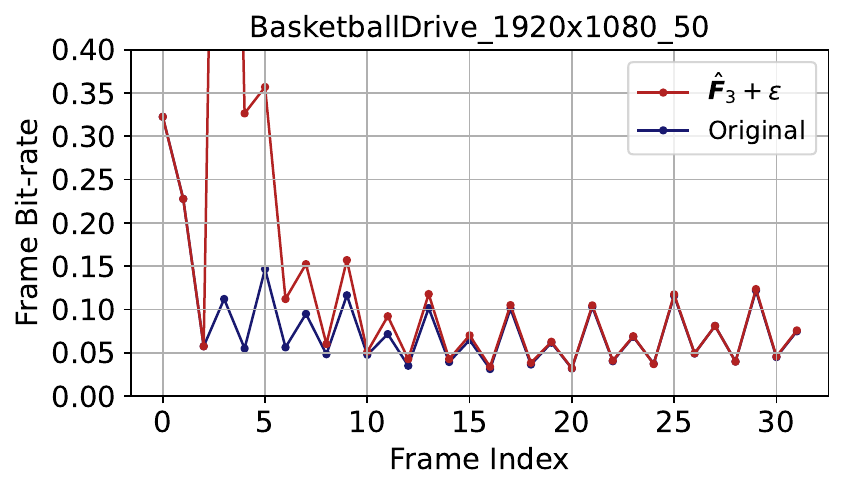}}
  \caption{PSNR / Bpp changes of DCVC-DC~\cite{li2023neural} when the propagated feature $\hat{\boldsymbol{F}}_3$ is with random Gaussian Noise $\boldsymbol{\varepsilon}\sim \mathcal{N}(0,1)$.}
  \label{fig:frame_impact}
\end{figure}
\begin{figure*}[ht]
  \centering
  \subfloat{
    \includegraphics[scale=0.38]{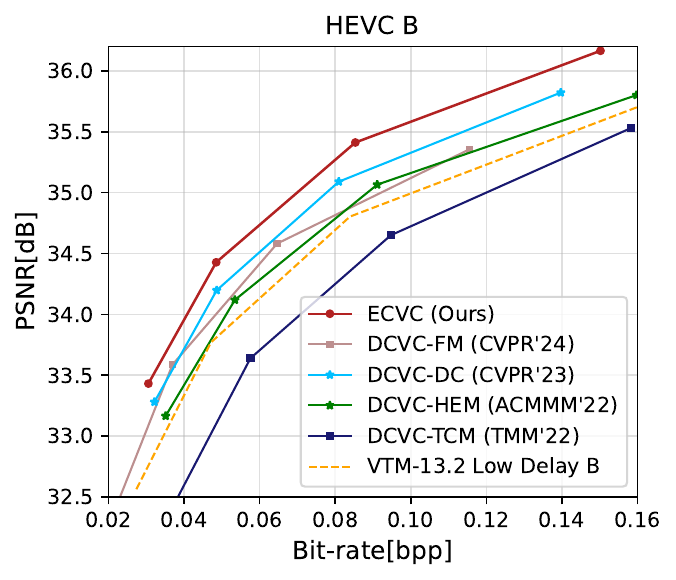}}
  \subfloat{
  \includegraphics[scale=0.38]{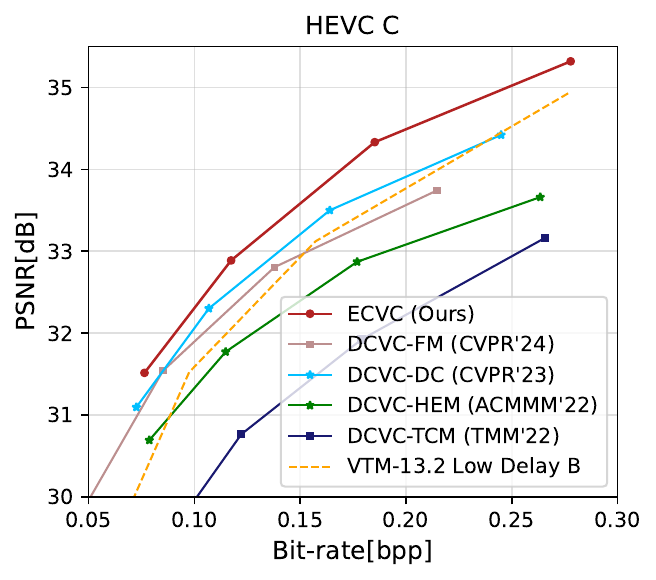}}
  \subfloat{
  \includegraphics[scale=0.38]{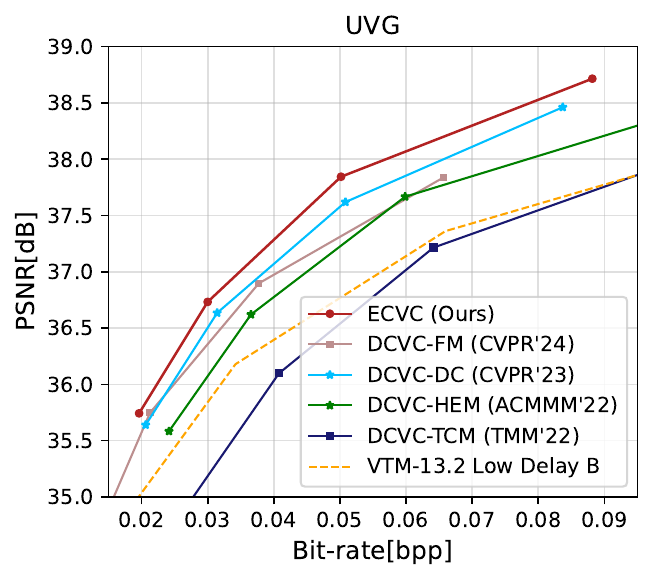}}
  \subfloat{
    \includegraphics[scale=0.38]{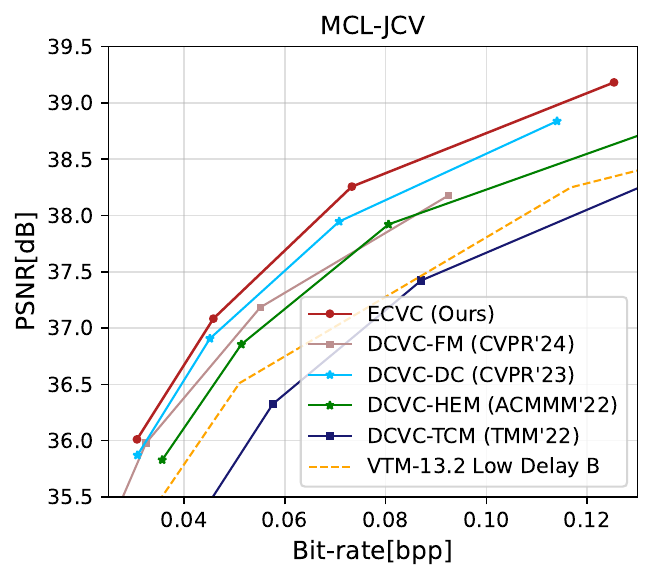}}
  \caption{Rate-Distortion curves on HEVC B, HEVC C, UVG and MCL-JCV dataset. \textbf{The intra period is $\bm{32}$ with $\bm{96}$ frames}.}
  \label{fig:rd_ip32}
\end{figure*}
\begin{table*}  
  \centering
  \setlength{\tabcolsep}{3.2mm}{
  \footnotesize
  \renewcommand\arraystretch{1.1}
    \begin{threeparttable}{
  \begin{tabular}{@{}cccccccccccccc@{}}
  \toprule
  \multicolumn{1}{c}{\multirow{2}{*}{Method}} & \multicolumn{1}{c|}{\multirow{2}{*}{{Venue}}}                            & \multicolumn{7}{c}{{BD-Rate (\%) w.r.t. VTM-13.2 LDB~\cite{bross2021overview}}}  \\
  \multicolumn{1}{c}{}      & \multicolumn{1}{c|}{}                                                & \multicolumn{1}{c}{HEVC B}    & \multicolumn{1}{c}{HEVC C}  & \multicolumn{1}{c}{HEVC D} & \multicolumn{1}{c}{HEVC E}  & \multicolumn{1}{c}{UVG~\cite{mercat2020uvg}} & \multicolumn{1}{c}{MCL-JCV~\cite{wang2016mcl}} & \multicolumn{1}{c}{Average} \\       \midrule                                      
  \multicolumn{1}{c}{DCVC-TCM~\cite{sheng2022temporal}}     & \multicolumn{1}{c|}{TMM'22}                                  & $+28.5$       & $+60.5$ & $+27.8$       & $+67.3$  & $+17.1$ & $+30.6$   & $+38.6$ \\
  \multicolumn{1}{c}{DCVC-HEM~\cite{li2022hybrid}}    & \multicolumn{1}{c|}{ACMMM'22}                                       & $-5.1$       & $+15.0$ & $-8.9$       & $+7.1$   & $-18.2$& $-6.4$ & $-2.8$ \\
  \multicolumn{1}{c}{DCVC-DC~\cite{li2023neural}}      & \multicolumn{1}{c|}{CVPR'23}                                     & ${-17.4}$       & $-9.8$ & ${-29.0}$       & $-26.0$  & ${-30.0}$  & ${-20.0}$ & ${-22.0}$ \\
  \multicolumn{1}{c}{DCVC-FM~\cite{li2024neural}}        & \multicolumn{1}{c|}{CVPR'24}                                   & $-12.5$       & ${-10.3}$ &  $-26.5$    & ${-26.9}$  & $-24.0$  & $-12.7$  & $-18.8$\\\midrule
  \multicolumn{1}{c}{ECVC }      & \multicolumn{1}{c|}{Ours}   & $\bm{-28.3}$     & $\bm{-19.6}$ & $\bm{-36.7}$      & $\bm{-27.1}$ & $\bm{-37.6}$ & $\bm{-26.3}$ & $\bm{-29.3}$   \\
  \bottomrule  
\end{tabular}}
\begin{tablenotes}    
  \footnotesize               
  \item[1] The quality indexes of DCVC-FM are set to match the bit-rate range of DCVC-DC. 
\end{tablenotes}       
\end{threeparttable}}
\caption{{BD-Rate $(\%)$~\cite{bjontegaard2001calculation} comparison for PSNR (dB). The anchor is \textbf{VTM-13.2 LDB}. \textbf{The Intra Period is $\bm{32}$ with $\bm{96}$ frames}.}}
\label{tab:rd}   
\end{table*}
\begin{table*}
  \centering
  \setlength{\tabcolsep}{3.2mm}{
  \footnotesize
  \renewcommand\arraystretch{1}
  \begin{threeparttable}{
  \begin{tabular}{@{}cccccccccccccc@{}}
  \toprule
  \multicolumn{1}{c}{\multirow{2}{*}{Method}} & \multicolumn{1}{c|}{\multirow{2}{*}{{Venue}}}                            & \multicolumn{7}{c}{{BD-Rate (\%) w.r.t. VTM-13.2 LDB~\cite{bross2021overview}}}  \\
  \multicolumn{1}{c}{}      & \multicolumn{1}{c|}{}                                                & \multicolumn{1}{c}{HEVC B}    & \multicolumn{1}{c}{HEVC C}  & \multicolumn{1}{c}{HEVC D} & \multicolumn{1}{c}{HEVC E} & \multicolumn{1}{c}{UVG~\cite{mercat2020uvg}} & \multicolumn{1}{c}{MCL-JCV~\cite{wang2016mcl}} & \multicolumn{1}{c}{Average} \\       \midrule             
  \multicolumn{1}{c}{DCVC-TCM~\cite{sheng2022temporal}}     & \multicolumn{1}{c|}{TMM'22}                                  & $-20.5$       & $-21.7$ & $-36.2$       & $-20.5$   & $-6.0$ & $-18.6$ & $-20.6$    \\
  \multicolumn{1}{c}{DCVC-HEM~\cite{li2022hybrid}}    & \multicolumn{1}{c|}{ACMMM'22}                                       & $-47.4$       & $-43.3$ & $-55.5$       & $-52.4$   & $-32.7$& $-44.0$ & $-45.9$ \\
  \multicolumn{1}{c}{DCVC-DC~\cite{li2023neural}}      & \multicolumn{1}{c|}{CVPR'23}                                     & ${-53.0}$       & ${-54.6}$ & ${-63.4}$       & ${-60.7}$  & ${-36.7}$  & ${-49.1}$ & ${-52.9}$ \\\midrule
  \multicolumn{1}{c}{ECVC }      & \multicolumn{1}{c|}{Ours}    & $\bm{-57.7}$     & $\bm{-58.2}$ & $\bm{-65.6}$      & $\bm{-60.5}$ & $\bm{-42.7}$ & $\bm{-54.9}$ & $\bm{-56.6}$  \\\bottomrule  
\end{tabular}}
\begin{tablenotes}    
  \footnotesize               
  \item[1] The MS-SSIM~\cite{wang2003multiscale} optimized weights of DCVC-FM are not open-sourced.   
\end{tablenotes}       
\end{threeparttable}}
\caption{{BD-Rate $(\%)$~\cite{bjontegaard2001calculation} comparison for MS-SSIM~\cite{wang2003multiscale}. The anchor is \textbf{VTM-13.2 LDB}. \textbf{The intra period is $\bm{32}$ with $\bm{96}$ frames}.}}
\label{tab:rd_ip32_ssim} 
\end{table*}
Increasing the frame count during training, however, results in escalated GPU memory consumption, posing a substantial computational burden.\par
To mitigate above issue, we propose a partial cascaded finetuning strategy (PCFS) following initial training on 6 frames.
\textit{In low-delay scenarios, neighboring frames have a greater impact on the current frame's coding than distant frames}. 
As illustrated in Figure~\ref{fig:frame_impact}, 
if the propagated feature $\hat{\boldsymbol{F}}_3$ is with Gaussian noise, the coding performances
of Frame $3\sim 13$ are significantly influenced, especially of the bit-rate while the 
the coding performances of frame $13\sim 32$ are less affected.
Thus, our approach involves partitioning the finetuning sequences into several groups. When a group is fed into our ECVC, 
the associated loss and the gradients are computed to update the model.
The PCFS process is illustrated in Figure~\ref{fig:pcfs}. The process is formulated as:
\begin{equation}
  \boldsymbol{\theta}_{j+1} = \boldsymbol{\theta}_j - \alpha \nabla \sum_{t=t_{j} +1}^{t_{(j+1)}} \mathcal{R}_t + \lambda \mathcal{D}_t,
\end{equation}
where $\boldsymbol{\theta}$ is the model weight, $j\in \mathbb{N}$ is the group index, $\alpha$ is the learning rate and $t_0 = 0$.
Empirically, this straightforward method yields effective results, 
despite the gradient of the first frame in each group not propagating to its previous frames. 
A potential extension could involve a shifted-window mechanism to enhance gradient propagation, 
although our results indicate no performance improvement over the original PCFS.
The PCFS alleviates error accumulation due to two aspects:
(1) fine-tuning on longer sequences reduces the mismatch between training and testing;
(2) the calculation of the loss within each group is cascaded, 
which makes the error propagation within the group aware during
fine-tuning owing to the effect of the gradients.
  
      
      
      
      
\section{Experiments}
\subsection{Experimental Setup}
The proposed ECVC is implemented with Pytorch 2.2.2~\cite{paszke2019pytorch} and trained with Vimeo-90K train split~\cite{xue2019video} and BVI-DVC dataset~\cite{ma2021bvi} with $4$ Tesla A100-80G GPUs.
Following the DCVC series~\cite{sheng2022temporal,li2023neural,li2022hybrid,li2024neural}, we apply the multi-stage training~\cite{sheng2022temporal} on Vimeo-90K.
The ECVC is further finetuned on BVI-DVC with the proposed partial cascaded training strategy. 
The sequences are randomly cropped to $256\times 256$ patches and
the batch size is $4$ during training and finetuning.
In the finetuning stage, $55$ frames are involved and divided into $3$ groups and the learning rate is $10^{-6}$.
The loss function is $\mathcal{L} = \mathcal{R} + \lambda \times \mathcal{D}$, where $\mathcal{R}$ is the bit-rate
and $\mathcal{D}$ is the distortion.
The $\lambda$ is set to $\{85,170,380,840\}$ for different bit-rates
when ECVC is optimized for MSE and the $\lambda$ is set to $\{7.68,15.36,30.72,61.44\}$
when ECVC is optimized for MS-SSIM~\cite{wang2003multiscale}.
We employ the intra frame codec of DCVC-DC~\cite{li2024neural} for intra frame coding. \par
Following existing literature~\cite{li2021deep,li2022hybrid,sheng2022temporal,li2024neural,li2023neural,sheng2024spatial,lu2024deep}, ECVC is evaluated on HEVC datasets~\cite{bossen2013common}, including 
class B, C, D, E, UVG~\cite{mercat2020uvg} and MCL-JCV~\cite{wang2016mcl}.
To fully demonstrate the superiority of the ECVC, we compare the ECVC with
DCVC-TCM~\cite{sheng2022temporal}, DCVC-HEM~\cite{li2022hybrid}, DCVC-DC~\cite{li2023neural}, and DCVC-FM~\cite{li2024neural}. The testing scenarios
are low delay with IP $32$ and IP $-1$~\cite{bossen2019jvet}. The distortion metric is PSNR and MS-SSIM~\cite{wang2003multiscale} in RGB color format.
\subsection{Comparisons with Previous SOTA Methods}
\textbf{Under IP $\bm{32}$}
The results under IP $32$ are presented in Table~\ref{tab:rd}, Table~\ref{tab:rd_ip32_ssim} and Figure~\ref{fig:rd_ip32}.
Our ECVC outperforms DCVC-DC and DCVC-FM on all datasets. Specifically, 
the ECVC achieves an average of $29.3\%$ bit-rate saving over the VTM-13.2. The bit-rate savings of DCVC-DC
and DCVC-FM are $22\%$ and $18.8\%$, respectively. The performance improvements over DCVC-DC demonstrate the
superiority of the proposed techniques.
In addition, regarding the MS-SSIM optimized models, ECVC achieves an average of $56.6\%$ bit-rate saving over the VTM-13.2. 
The proposed ECVC outperforms the DCVC-DC by reducing $3.7\%$ more bit-rate.
\begin{figure*}[ht]
  \centering
  \subfloat{
    \includegraphics[scale=0.4]{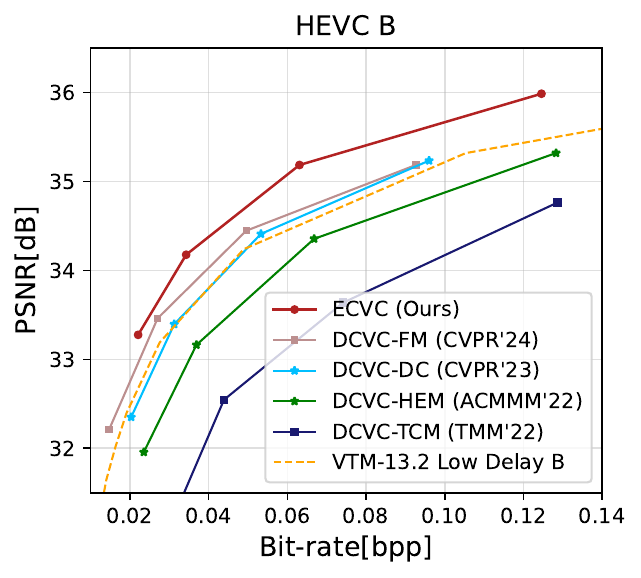}}
  \subfloat{
    \includegraphics[scale=0.4]{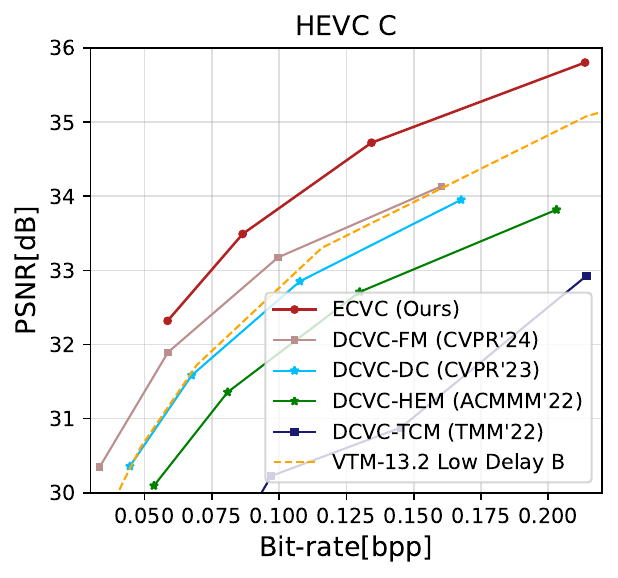}}
  \subfloat{
  \includegraphics[scale=0.4]{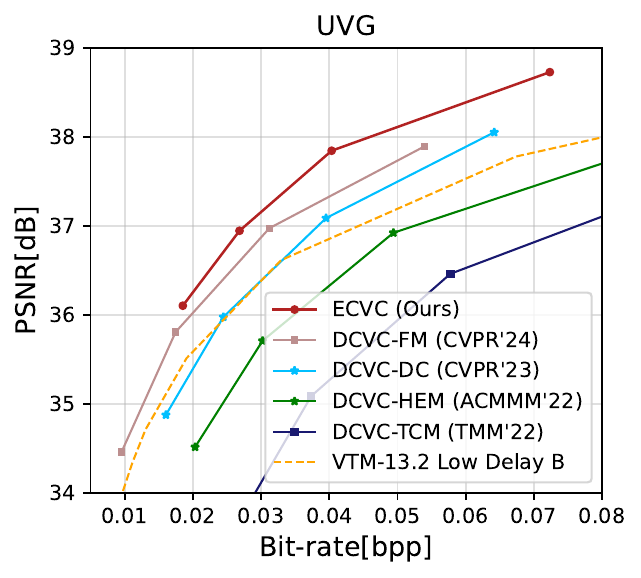}}
  \subfloat{
    \includegraphics[scale=0.4]{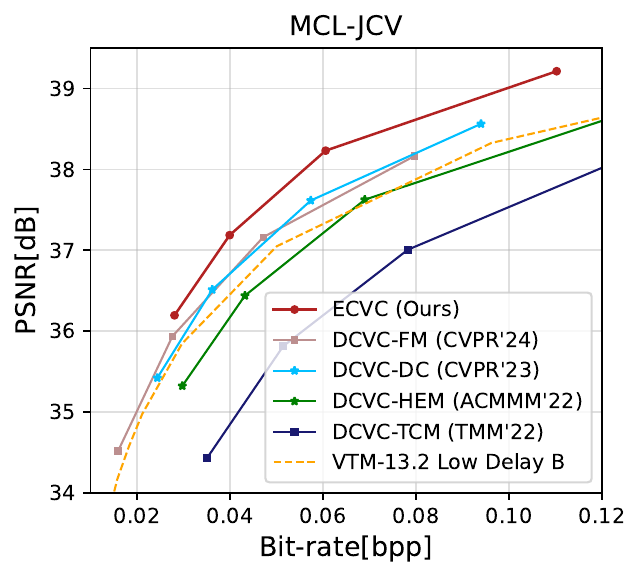}}
  \caption{Rate-Distortion curves on HEVC B, HEVC C, UVG and MCL-JCV dataset. \textbf{The intra period is $\bm{-1}$ with All Frames}.}
  \label{fig:rd_ip-1}
\end{figure*}
\begin{table*}[ht]
  \centering
  \footnotesize
  \setlength{\tabcolsep}{3.2mm}{
  \renewcommand\arraystretch{1}
  \begin{tabular}{@{}cccccccccccccc@{}}
  \toprule
  \multicolumn{1}{c}{\multirow{2}{*}{Method}} & \multicolumn{1}{c|}{\multirow{2}{*}{{Venue}}}                            & \multicolumn{7}{c}{{BD-Rate (\%) w.r.t. VTM-13.2 LDB~\cite{bross2021overview}}}  \\
  \multicolumn{1}{c}{}      & \multicolumn{1}{c|}{}                                                & \multicolumn{1}{c}{HEVC B}    & \multicolumn{1}{c}{HEVC C}  & \multicolumn{1}{c}{HEVC D} & \multicolumn{1}{c}{HEVC E} & \multicolumn{1}{c}{UVG~\cite{mercat2020uvg}} & \multicolumn{1}{c}{MCL-JCV~\cite{wang2016mcl}} & \multicolumn{1}{c}{Average} \\       \midrule                                      
  \multicolumn{1}{c}{DCVC-TCM~\cite{sheng2022temporal}}    & \multicolumn{1}{c|}{TMM'22}                                       & $+55.4$       & $+97.4$ & $+50.0$       & $+214.2$   & $+60.4$& $+50.7$ & $+88.0$ \\
  \multicolumn{1}{c}{DCVC-HEM~\cite{li2022hybrid}}    & \multicolumn{1}{c|}{ACMMM'22}                                       & $+3.9$       & $+28.4$ & $-1.2$       & $+66.3$   & $+0.5$& $+1.7$ & $+16.6$ \\
  \multicolumn{1}{c}{DCVC-DC~\cite{li2023neural}}        & \multicolumn{1}{c|}{CVPR'23}                                   & $-11.0$       & $+0.2$ &  $-23.9$    & $-7.8$  & $-21.0$  & $-13.0$  & $-12.8$\\
  \multicolumn{1}{c}{DCVC-FM~\cite{li2024neural}}        & \multicolumn{1}{c|}{CVPR'24}                                   & $-11.7$       & $-7.9$ &  $-28.2$    & $-25.8$  & $-23.9$  & $-12.3$  & $-18.3$\\\midrule
  \multicolumn{1}{c}{ECVC }      & \multicolumn{1}{c|}{Ours}   & $\bm{-27.9}$    & $\bm{-18.9}$ &$\bm{-39.0}$  &$\bm{-26.4}$      & $\bm{-38.3}$ & $\bm{-27.7}$ & $\bm{-29.7}$   \\
  \bottomrule  
\end{tabular}}
\caption{{BD-Rate $(\%)$~\cite{bjontegaard2001calculation} comparison for PSNR (dB). The anchor is \textbf{VTM-13.2 LDB}. \textbf{The Intra Period is $\bm{-1}$ with $\bm{96}$ frames.}}}
\label{tab:rd_ip-1} 
\end{table*}
\begin{table*}
  \centering
  \footnotesize
  \setlength{\tabcolsep}{3.7mm}{
  \renewcommand\arraystretch{1}
  \begin{tabular}{@{}cccccccccccccc@{}}
  \toprule
  \multicolumn{1}{c}{\multirow{2}{*}{Method}} & \multicolumn{1}{c|}{\multirow{2}{*}{{Venue}}}                            & \multicolumn{7}{c}{{BD-Rate (\%) w.r.t. VTM-13.2 LDB~\cite{bross2021overview}}}  \\
  \multicolumn{1}{c}{}      & \multicolumn{1}{c|}{}                                                & \multicolumn{1}{c}{HEVC B}    & \multicolumn{1}{c}{HEVC C}  & \multicolumn{1}{c}{HEVC D} & \multicolumn{1}{c}{HEVC E} & \multicolumn{1}{c}{UVG} & \multicolumn{1}{c}{MCL-JCV} & \multicolumn{1}{c}{Average} \\       \midrule                                      
  \multicolumn{1}{c}{DCVC-TCM~\cite{sheng2022temporal}}    & \multicolumn{1}{c|}{TMM'22}                                       & $+107.3$       & $+143.5$ & $+99.2$       & $+835.9$   & $+120.6$& $+63.7$ & $+228.4$ \\
  \multicolumn{1}{c}{DCVC-HEM~\cite{li2022hybrid}}    & \multicolumn{1}{c|}{ACMMM'22}                                       & $+22.8$       & $+32.3$ & $+13.4$       & $+236.9$   & $+33.5$& $+6.7$ & $+57.6$ \\
  \multicolumn{1}{c}{DCVC-DC~\cite{li2023neural}}        & \multicolumn{1}{c|}{CVPR'23}                                   & $-7.5$       & $+3.4$ &  $-12.0$    & $+83.9$  & $-4.5$  & $-12.9$  & $+8.4$\\
  \multicolumn{1}{c}{DCVC-FM~\cite{li2024neural}}        & \multicolumn{1}{c|}{CVPR'24}                                   & $-19.9$       & $-17.4$ &  $-25.7$    & $-24.5$  & $-22.5$  & $-13.4$  & $-20.6$\\\midrule
  \multicolumn{1}{c}{ECVC }      & \multicolumn{1}{c|}{Ours}   & $\bm{-33.4}$    & $\bm{-29.5}$ &$\bm{-38.8}$  &$\bm{-23.5}$      & $\bm{-37.5}$ & $\bm{-29.7}$ & $\bm{-32.1}$   \\
  \bottomrule  
\end{tabular}}
\caption{{BD-Rate $(\%)$~\cite{bjontegaard2001calculation} comparison for PSNR (dB). The anchor is \textbf{VTM-13.2 LDB}. \textbf{The Intra Period is $\bm{-1}$ with All frames.}}}
\label{tab:rd_ip-1_all} 
\end{table*}
\begin{table}  
  \centering
  \footnotesize
  \setlength{\tabcolsep}{2mm}{
  \renewcommand\arraystretch{1}
  \begin{tabular}{@{}cccccccccccccc@{}}
  \toprule
  \multicolumn{1}{c}{\multirow{1}{*}{Method}}                            & \multicolumn{1}{c}{Params (M)} & \multicolumn{1}{c}{kMACs/pixel} & \multicolumn{1}{c}{ET(s)} & \multicolumn{1}{c}{DT(s)}\\ \midrule                                      
  \multicolumn{1}{c}{DCVC-HEM~\cite{li2022hybrid}}                                       & $50.9$       & $1581$ & $0.67$       & $0.52$ \\
  \multicolumn{1}{c}{DCVC-DC~\cite{li2023neural}}                                       & $50.8$       & $1274$ & $0.74$       & $0.59$ \\
  \multicolumn{1}{c}{DCVC-FM~\cite{li2024neural}}                                         & $44.9$       & $1073$ &  $0.73$    & $0.60$  \\\midrule
  \multicolumn{1}{c}{ECVC }      & 61.9    & $1407$ &  $0.78$    & $0.62$ \\\bottomrule  
\end{tabular}}
\caption{{Complexity comparison among proposed ECVC, DCVC-HEM, DCVC-DC, and DCVC-FM.
``Params" denotes the number of model parameters. ``KMACs/pixel" denotes the multiply-add operations per pixel on 1080p sequences. ``ET", ``DT" are average encoding and decoding time per frame.}}
\label{tab:complex} 
\end{table}
\par
\textbf{Under IP $\bm{-1}$}
The results under IP -1 are presented in Table~\ref{tab:rd_ip-1}, Table~\ref{tab:rd_ip-1_all} and Figure~\ref{fig:rd_ip-1}.
Following DCVC-FM~\cite{li2024neural}, the performance under IP -1 is evaluated on $96$ frames and \textit{all} frames.
When evaluated on $96$ frames, our ECVC reduces $11.1\%$ more bit-rate over VTM-13.2 compared to DCVC-FM.
When evaluated on \textit{all} frames, our ECVC reduces $11.5\%$ more bit-rate over VTM-13.2 compared to DCVC-FM.
The significant performance improvement of our ECVC over DCVC-DC and 
DCVC-FM demonstrates the effectiveness of the proposed techniques.
\par
\textbf{Subjective quality comparison}
The subjective quality comparison is presented in Figure~\ref{fig:vis}. 
Compared with DCVC-FM, the ECVC consumes lower bit-rates and the reconstruction frame achieves $1$ dB improvements in PSNR.
In terms of subjective quality, our ECVC has significant improvements compared to previous SOTA methods, such as DCVC-DC and DCVC-FM.
\subsection{Complexity Analysis}
We compare the complexity of ECVC with recent DCVC-HEM and DCVC-FM and the baseline DCVC-DC.
The results are presented in Table~\ref{tab:complex}.
Since the ECVC is based on DCVC-DC with the involvement of MNLC, the complexity of ECVC is slightly higher than that of DCVC-DC. 
Considering the rate-distortion performance advancement of ECVC over DCVC-DC, the introduction of MNLC is worthwhile.
\begin{table}[t]   
  \centering
  \footnotesize
  \setlength{\tabcolsep}{0.35mm}{
  \renewcommand\arraystretch{1}
  \begin{tabular}{@{}cccccccccccccc@{}}
    \toprule
    \multicolumn{1}{c|}{{Methods}} &\multicolumn{1}{c|}{{IP}}                           & \multicolumn{1}{c}{B}  & \multicolumn{1}{c}{C}& \multicolumn{1}{c}{D} & \multicolumn{1}{c}{E} & \multicolumn{1}{c}{Avg}\\ \midrule                                      
    \multicolumn{1}{c|}{{Base (DCVC-DC) Large}  }      & \multicolumn{1}{c|}{{$32$}}         & {$-2.1$}  & {$-4.2$}  & {$-5.4$} & {$-0.2$} &{$-3.0$}\\
    \multicolumn{1}{c|}{{Base + MNLC}}   & \multicolumn{1}{c|}{{$32$}}      & {\bm{$-4.0$}} & {\bm{$-9.9$}} & {\bm{$-9.0$}} & {\bm{$-9.7$}}  &{\bm{$-8.2$}}   \\\midrule   
    \multicolumn{1}{c|}{{Base (DCVC-DC) Large }}      & \multicolumn{1}{c|}{{$-1$}}         & {$-1.4$}  & {$-2.4$}  & {$-4.3$} & {$-4.7$} &{$-3.2$}\\
    \multicolumn{1}{c|}{{Base + MNLC}}   & \multicolumn{1}{c|}{{$-1$}}      & {\bm{$-2.5$}} & {\bm{$-6.5$}} & {\bm{$-7.9$}} & {\bm{$-16.9$}}  &{\bm{$-8.5$}}   \\\midrule 
    \multicolumn{1}{c|}{DCVC-HEM}      & \multicolumn{1}{c|}{$32$}         & $23.1$  &$23.3$   &$26.0$ &$22.9$&$23.8$  \\
    \multicolumn{1}{c|}{DCVC-HEM + PCFS}      & \multicolumn{1}{c|}{$32$}         & $14.1$  & $10.6$  & $14.5$& $18.7$ &$14.5$\\
    \multicolumn{1}{c|}{Base (DCVC-DC) + PCFS}      & \multicolumn{1}{c|}{$32$}         & $-1.9$  & $-8.5$  & $-8.6$ & $-8.0$ &$-6.8$\\
    \multicolumn{1}{c|}{Base + PCFS + NLC}      & \multicolumn{1}{c|}{$32$}         &  $-4.6$  & $-13.0$  &$-12.1$ & $-11.4$ &$-10.3$\\
    \multicolumn{1}{c|}{Base + PCFS + MNLC}   & \multicolumn{1}{c|}{$32$}      & \bm{$-9.1$} & \bm{$-15.3$} & \bm{$-14.5$} & \bm{$-14.8$}  &\bm{$-13.4$}   \\\midrule   
    \multicolumn{1}{c|}{DCVC-HEM}      & \multicolumn{1}{c|}{$-1$}         & $29.8$   & $30.8$  &$32.9$ & $18.1$ & $27.9$\\
    \multicolumn{1}{c|}{DCVC-HEM + PCFS}      & \multicolumn{1}{c|}{$-1$}         &$9.8$   & $5.9$  &$9.7$ & $-7.4$ & $4.5$  \\
    \multicolumn{1}{c|}{Base (DCVC-DC) + PCFS}      & \multicolumn{1}{c|}{$-1$}         & $-6.0$  & $-11.8$  & $-13.2$&$-19.4$ &$-12.6$    \\
    \multicolumn{1}{c|}{Base + PCFS + NLC }      & \multicolumn{1}{c|}{$-1$}        & $-10.4$  & $-18.5$  &$-19.5$ &$-27.4$ &$-19.0$ \\
    \multicolumn{1}{c|}{Base + PCFS + MNLC }   & \multicolumn{1}{c|}{$-1$}      & \bm{$-18.1$} & \bm{$-25.7$} & \bm{$-24.8$}  & \bm{$-45.1$} &\bm{$-28.4$}\\\bottomrule  
  \end{tabular}}
\caption{{Ablation Studies on HEVC B, C, D, E. The anchor is the base model. ``IP" denotes the intra period.}}
\label{tab:ablation}   
\end{table}
\begin{table}
  \centering
  \footnotesize
  \setlength{\tabcolsep}{2mm}{
  \renewcommand\arraystretch{1}
  \begin{tabular}{@{}cccccccccccccc@{}}
  \toprule
  \multicolumn{1}{c|}{\multirow{1}{*}{Frames}} &\multicolumn{1}{c|}{\multirow{1}{*}{IP}}                           & \multicolumn{1}{c}{B}  & \multicolumn{1}{c}{C}& \multicolumn{1}{c}{D} & \multicolumn{1}{c}{E} \\ \midrule                                      
  \multicolumn{1}{c|}{$6$}      & \multicolumn{1}{c|}{$32$}     & $-4.0$ & $-9.9$  & $-9.0$ & $-9.7$      \\
  \multicolumn{1}{c|}{$20$}        & \multicolumn{1}{c|}{$32$}     & $-8.6$ & $-14.5$  & $-12.4$ & $-12.1$      \\
  \multicolumn{1}{c|}{$38$}   & \multicolumn{1}{c|}{$32$}    & $-8.9$ & $-15.0$  & $-14.4$ & $-13.8$      \\
  \multicolumn{1}{c|}{$55$}      & \multicolumn{1}{c|}{$32$}     & $-9.1$ & $-15.3$ & $-14.5$ & $-14.8$    \\\midrule 
  \multicolumn{1}{c|}{$6$}      & \multicolumn{1}{c|}{$-1$}      & $-0.4$  & $-2.9$ & $-6.7$  & $-9.5$      \\
  \multicolumn{1}{c|}{$20$}        & \multicolumn{1}{c|}{$-1$}             & $-17.1$  & $-24.4$ & $-22.3$  & $-41.3$      \\
  \multicolumn{1}{c|}{$38$}   & \multicolumn{1}{c|}{$-1$}         & $-18.1$  & $-25.7$ & $-24.2$  & $-42.5$      \\
  \multicolumn{1}{c|}{$55$}      & \multicolumn{1}{c|}{$-1$}      & $-18.8$ & $-26.1$ & $-24.8$  & $-45.1$      \\
  \bottomrule  
\end{tabular}}
\caption{{Influences of finetuning frames. The anchor is the base model. ``IP" denotes the intra period.}}
\label{tab:frame_ablation}   
\end{table}
\begin{figure}[t]
  \centering
  \subfloat{
    \includegraphics[scale=0.135]{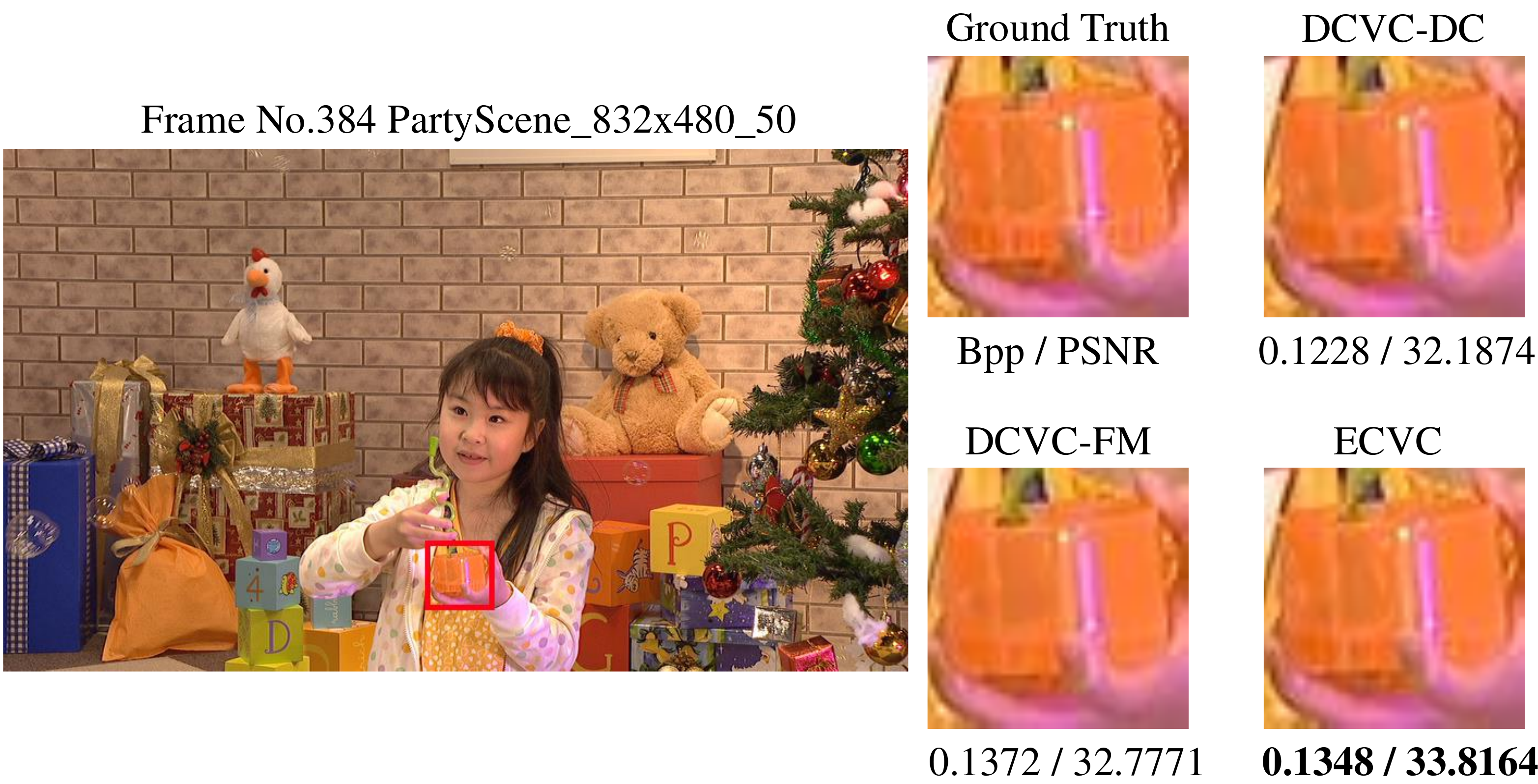}}
  \caption{Subjective quality comparison on reconstruction frames of DCVC-DC~\cite{li2023neural}, DCVC-FM~\cite{li2024neural}, the proposed ECVC, and the ground truth.}
  \label{fig:vis}
\end{figure}
\subsection{Ablation Studies}
In ablation studies, all models are optimized with MSE, and PSNR is adopted to evaluate distortion.
Table~\ref{tab:ablation} presents the improvements of
each component under IP $32$ and IP $-1$ with $96$ frames setting.
``Base" is our reproduced baseline DCVC-DC$^\ast$.
``Base Large" has 60M params by proportionally increasing all channel dimensions.
``LNC" denotes only one reference frame for local and non-local context mining (meaning Base + one reference frame non-local context).\par
\textbf{Under IP $\bm{32}$}
LNC is able to capture non-local correlations, 
which effectively improves the rate-distortion performance.
Compared with LNC, the MNLC adopts one more reference frame,
enhancing the performance.
Thanks to the proposed PCFS, the mismatch of frames between 
training and inference can be alleviated.
The error accumulation is reduced and the 
performance is further improved.\par
\textbf{Under IP $\bm{-1}$}
Under a long prediction chain, the influence of PCFS is more significant,
due to the impacts of error accumulation.
Specifically, when equipped with PCFS, the model reduces $45.1\%$ bit-rate 
over the baseline on HEVC E. Our proposed PCFS is able to finetune the model
on \textit{much longer} sequences, thereby reducing accumulated errors.\par
\begin{figure}[t]
  \centering
  \subfloat{
    \includegraphics[scale=0.34]{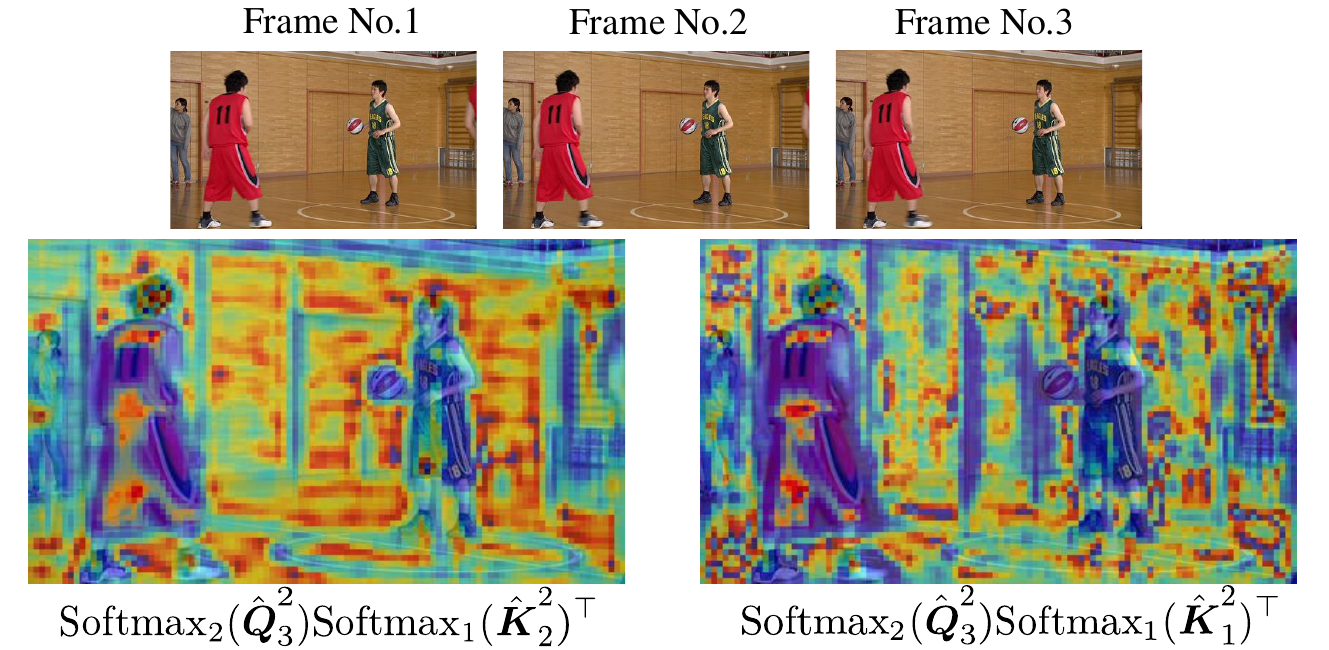}}
  \caption{Visualization of non-local correlations. The inputs are the first 3 frames in ``BasketballPass\_416x240\_50".}
  \label{fig:nlc}
\end{figure}
\textbf{Visualization of captured non-local correlations}
To demonstrate the effectiveness of the proposed MNLC in capturing non-local correlations, we visualize its attention maps in Figure~\ref{fig:nlc}.
For each visualization, we select a query and compute its attention scores with all keys.
Redder colors indicate higher attention scores.
It is obvious the MNLC captures the non-local contexts between walls and floors (the texture of walls and floors are similar).
The captured non-local contexts are employed as priors for conditional encoding, thereby saving the bit-rate for inter coding.\par
\textbf{Influences of finetuning frames}
With the proposed PCFS, the model can be finetuned with long sequences under limited computational resources.
To further analyze the PCFS,
we compare the performance of fine-tuning with sequences of different lengths.
The results are reported in Table~\ref{tab:frame_ablation}.
The frame numbers are $\{6,20,38,55\}$ with $\{1,1,2,3\}$ groups, respectively.
When involving more frames for finetuning, the performance is further improved under both IP $-1$ and IP $32$ scenarios.
Notably, for sequences with smaller movements (HEVC E~\cite{bossen2013common}), the gain of using more frames for fine-tuning is substantial.
The performance improvements under both IP $-1$ and IP $32$ settings as the frame number 
increases can be attributed to the diverse error patterns encountered during fine-tuning, 
which enhance ECVC's generalization capability.
\section{Conclusion}
In this paper, we demonstrate the effectiveness of exploiting 
non-local correlations for learned video compression.
To extract more temporal priors from multiple frames, we propose the Multiple Frame Non-Local Context Mining approach. The offset diversity, successive flow warping and multi-scale refinement are employed to capture
local correlations across multiple frames, while multi-head linear cross attention is
employed to capture non-local correlations among them.
To reduce the temporal error accumulation, we introduce the partial cascaded finetuning 
strategy to optimize the model under limited resources.
However, while ECVC achieves SOTA performance, it may perform worse than 
VTM on out-of-domain sequences (\textit{e.g.}, anime videos) because ECVC is trained on natural videos~\cite{xue2019video,ma2021bvi}. 
To address this issue, we will investigate the 
instance-adaptive optimization techniques~\cite{vaninstance,xu2023bit,tang2024offline,lu2020content} in the future.
{
    \small
    \bibliographystyle{ieeenat_fullname}
    \bibliography{main}
}
\clearpage
\appendix

\section{Network Structure}
Our ECVC is based on DCVC-DC~\cite{li2023neural}, but focuses more on exploiting non-local correlations in multiple frames
to boost the rate-distortion performance.
Here we describe 
 implementation details in the introduced components to DCVC-DC.
\subsection{Embedding Layers and Depth-wise Res Block}
In the multi-head linear cross attention layer, the depth-wise res block are employed for embedding and point-wise interactions.
The structure of the depth-wise res block is depicted in Figure~\ref{fig:dwrb}.
\subsection{Multi-Scale Refine Module}
The architecture of multi-scale refine module is depicted in Figure~\ref{fig:mlrf}.
\begin{figure}
    \centering
    \subfloat{
      \includegraphics[scale=0.31]{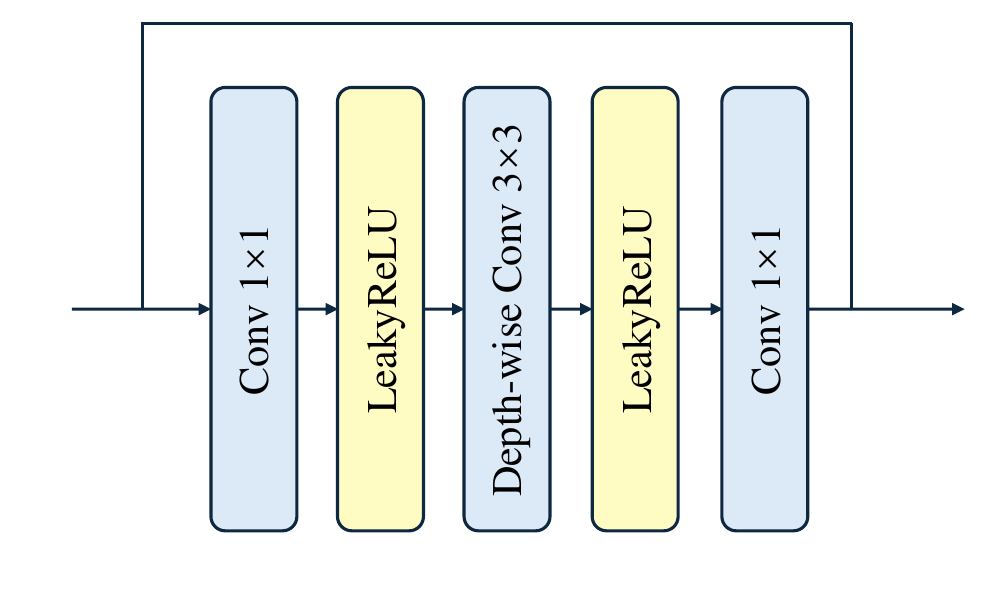}}
    \caption{Architecture of the depth-wise Res Block.}
    \label{fig:dwrb}
  \end{figure}
  \begin{figure*}
    \centering
    \subfloat{
      \includegraphics[scale=0.44]{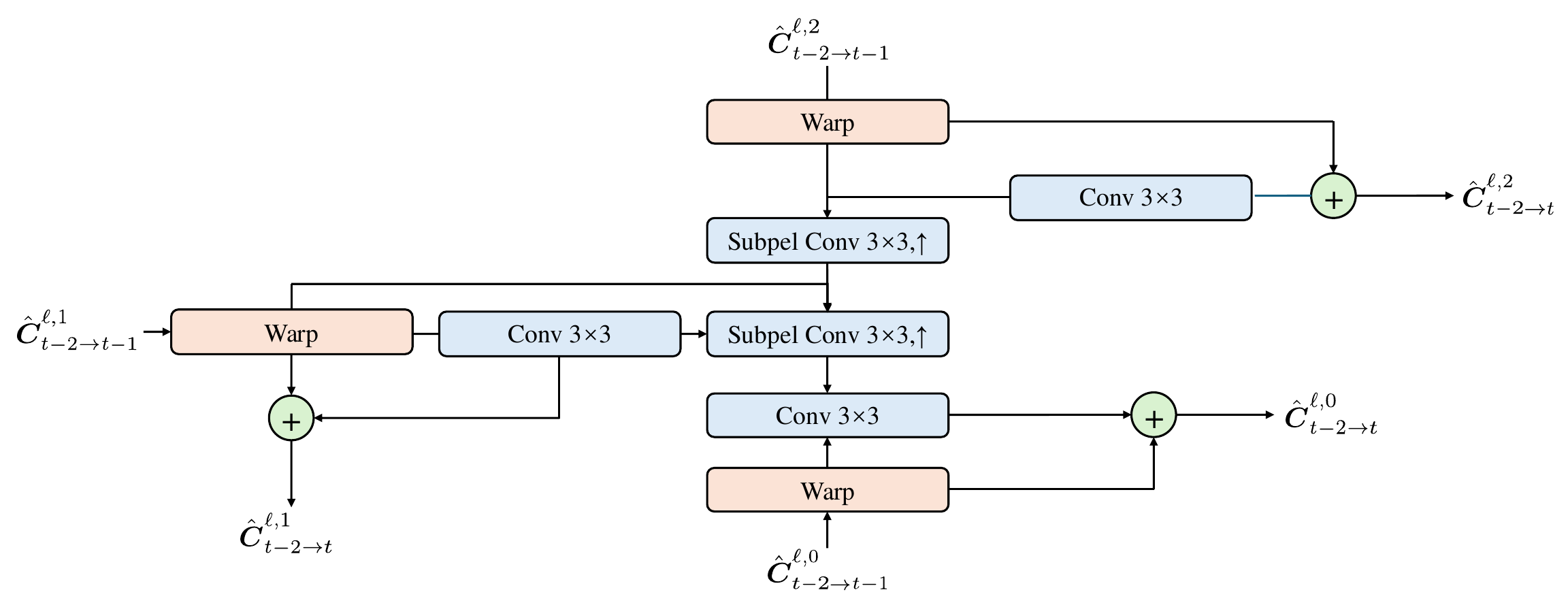}}
    \caption{Architecture of multi-scale refine module. ``Subpel Conv'' means sub pixel convolution.}
    \label{fig:mlrf}
  \end{figure*}
\subsection{Channel numbers}
The channel numbers of multi-scale features are $\{d^0_f, d^1_f, d^2_f\} = \{48, 64, 96\}$.\par
The channel numbers of mid-features during encoding are $\{d^0_e, d^1_e, d^2_e\} = \{3, 64, 96\}$.\par
The channel numbers of mid-features during decoding are $\{d^0_d, d^1_d, d^2_d\} = \{32, 64, 96\}$.
\section{ECVC-FM}
We introduce ECVC-FM, a new variant of ECVC that adopts DCVC-FM~\cite{li2024neural} as its backbone while incorporating key techniques from ECVC.
Additionally, we modify the channel dimensions of temporal contexts, resulting in a model with $53.05$M parameters.
\section{Test Settings}
To conduct comparisons, 
we compare the LVCs and traditional codecs in both RGB color space.
BT.601 is employed to convert the frames in YUV color space to RGB color space.
To obtain the RD data of VTM-13.2 LDB~\cite{bross2021overview}, following commands are employed:
\noindent\begin{minipage}{\linewidth}
    \begin{verbatim}
    EncoderApp 
    -c encoder_lowdelay_vtm.cfg 
    --InputFile={input file name} 
    --BitstreamFile={bitstream file name} 
    --DecodingRefreshType=2 
    --InputBitDepth=8 
    --OutputBitDepth=8 
    --OutputBitDepthC=8 
    --InputChromaFormat=444 
    --FrameRate={frame rate} 
    --FramesToBeEncoded={frame number} 
    --SourceWidth={width} 
    --SourceHeight={height}
    --IntraPeriod={IP}
    --QP={qp} 
    --Level=6.2
    \end{verbatim}
    \end{minipage}\par
For DCVC-DC and DCVC-FM~\cite{li2024neural}, we use the official weights and code to evaluate the performance.
\begin{table*}[ht]   
  \centering
  \setlength{\tabcolsep}{2.6mm}{
  \small
  \renewcommand\arraystretch{1.1}
    \begin{threeparttable}{
  \begin{tabular}{@{}cccccccccccccc@{}}
  \toprule
  \multicolumn{1}{c}{\multirow{2}{*}{Method}} & \multicolumn{1}{c|}{\multirow{2}{*}{{Venue}}}                            & \multicolumn{7}{c}{{BD-Rate (\%) w.r.t. VTM-13.2 LDB~\cite{bross2021overview}}}  \\
  \multicolumn{1}{c}{}      & \multicolumn{1}{c|}{}                                                & \multicolumn{1}{c}{HEVC B}    & \multicolumn{1}{c}{HEVC C}  & \multicolumn{1}{c}{HEVC D} & \multicolumn{1}{c}{HEVC E}  & \multicolumn{1}{c}{UVG~\cite{mercat2020uvg}} & \multicolumn{1}{c}{MCL-JCV~\cite{wang2016mcl}} & \multicolumn{1}{c}{Average} \\       \midrule                                      
  \multicolumn{1}{c}{DCVC-TCM~\cite{sheng2022temporal}}     & \multicolumn{1}{c|}{TMM'22}                                  & $+28.5$       & $+60.5$ & $+27.8$       & $+67.3$  & $+17.1$ & $+30.6$   & $+38.6$ \\
  \multicolumn{1}{c}{DCVC-HEM~\cite{li2022hybrid}}    & \multicolumn{1}{c|}{ACMMM'22}                                       & $-5.1$       & $+15.0$ & $-8.9$       & $+7.1$   & $-18.2$& $-6.4$ & $-2.8$ \\
  \multicolumn{1}{c}{DCVC-DC~\cite{li2023neural}}      & \multicolumn{1}{c|}{CVPR'23}                                     & ${-17.4}$       & $-9.8$ & ${-29.0}$       & $-26.0$  & ${-30.0}$  & ${-20.0}$ & ${-22.0}$ \\
  \multicolumn{1}{c}{DCVC-FM~\cite{li2024neural}}        & \multicolumn{1}{c|}{CVPR'24}                                   & $-12.5$       & ${-10.3}$ &  $-26.5$    & ${-26.9}$  & $-24.0$  & $-12.7$  & $-18.8$\\\midrule
  \multicolumn{1}{c}{DCVC-DC$^{*}$ }      & \multicolumn{1}{c|}{Reproduced}   & $-20.5$  & $-6.8$ & $-27.1$ & $-12.1$ & $-29.0$ & $-19.3$ & $-19.1$\\
  \multicolumn{1}{c}{ECVC }      & \multicolumn{1}{c|}{Ours}   & $\bm{-28.3}$     & $\bm{-19.6}$ & $\bm{-36.7}$      & $\bm{-27.1}$ & $\bm{-37.6}$ & $\bm{-26.3}$ & $\bm{-29.3}$   \\
  \multicolumn{1}{c}{ECVC-FM}      & \multicolumn{1}{c|}{Ours}   & $\bm{-30.9}$     & $\bm{-22.9}$ & $\bm{-40.7}$      & $\bm{-23.4}$ & $\bm{-42.5}$ & $\bm{-33.6}$ & $\bm{-32.3}$   \\\bottomrule  
\end{tabular}}
\begin{tablenotes}    
  \footnotesize               
  \item[1] The quality indexes of DCVC-FM are set to match the bit-rate range of DCVC-DC.
  \item[2] Please note that the ECVC is based on our reproduced DCVC-DC$^{*}$ since training scripts of DCVC series are not open-sourced.
\end{tablenotes}       
\end{threeparttable}}
\caption{{BD-Rate $(\%)$~\cite{bjontegaard2001calculation} comparison for PSNR (dB). The anchor is \textbf{VTM-13.2 LDB}. \textbf{The Intra Period is $\bm{32}$ with $\bm{96}$ frames}.}}
\label{tab:supp_rd}   
\end{table*}
\begin{table*}
  \centering
  \setlength{\tabcolsep}{2.5mm}{
  \small
  \renewcommand\arraystretch{1}
  \begin{threeparttable}{
  \begin{tabular}{@{}cccccccccccccc@{}}
  \toprule
  \multicolumn{1}{c}{\multirow{2}{*}{Method}} & \multicolumn{1}{c|}{\multirow{2}{*}{{Venue}}}                            & \multicolumn{7}{c}{{BD-Rate (\%) w.r.t. VTM-13.2 LDB~\cite{bross2021overview}}}  \\
  \multicolumn{1}{c}{}      & \multicolumn{1}{c|}{}                                                & \multicolumn{1}{c}{HEVC B}    & \multicolumn{1}{c}{HEVC C}  & \multicolumn{1}{c}{HEVC D} & \multicolumn{1}{c}{HEVC E} & \multicolumn{1}{c}{UVG~\cite{mercat2020uvg}} & \multicolumn{1}{c}{MCL-JCV~\cite{wang2016mcl}} & \multicolumn{1}{c}{Average} \\       \midrule             
  \multicolumn{1}{c}{DCVC-TCM~\cite{sheng2022temporal}}     & \multicolumn{1}{c|}{TMM'22}                                  & $-20.5$       & $-21.7$ & $-36.2$       & $-20.5$   & $-6.0$ & $-18.6$ & $-20.6$    \\
  \multicolumn{1}{c}{DCVC-HEM~\cite{li2022hybrid}}    & \multicolumn{1}{c|}{ACMMM'22}                                       & $-47.4$       & $-43.3$ & $-55.5$       & $-52.4$   & $-32.7$& $-44.0$ & $-45.9$ \\
  \multicolumn{1}{c}{DCVC-DC~\cite{li2023neural}}      & \multicolumn{1}{c|}{CVPR'23}                                     & ${-53.0}$       & ${-54.6}$ & ${-63.4}$       & ${-60.7}$  & ${-36.7}$  & ${-49.1}$ & ${-52.9}$ \\\midrule
  \multicolumn{1}{c}{DCVC-DC$^{*}$}      & \multicolumn{1}{c|}{Reproduced} & ${-53.5}$  & ${-52.8}$ & ${-61.6}$ & ${-48.8}$ & ${-39.1}$ & ${-52.2}$& ${-51.3}$\\
  \multicolumn{1}{c}{ECVC }      & \multicolumn{1}{c|}{Ours}    & $\bm{-57.7}$     & $\bm{-58.2}$ & $\bm{-65.6}$      & $\bm{-60.5}$ & $\bm{-42.7}$ & $\bm{-54.9}$ & $\bm{-56.6}$ 
  \\\bottomrule  
\end{tabular}}
\begin{tablenotes}    
  \footnotesize               
  \item[1] The MS-SSIM~\cite{wang2003multiscale} optimized weights of DCVC-FM are not open-sourced.   
\end{tablenotes}       
\end{threeparttable}}
\caption{{BD-Rate $(\%)$~\cite{bjontegaard2001calculation} comparison for MS-SSIM~\cite{wang2003multiscale}. The anchor is \textbf{VTM-13.2 LDB}. \textbf{The intra period is $\bm{32}$ with $\bm{96}$ frames}.}}
\label{tab:supp_rd_ip32_ssim} 
\end{table*}
\begin{table*}[ht]
  \centering
  \small
  \setlength{\tabcolsep}{2.5mm}{
  \renewcommand\arraystretch{1}
  \begin{tabular}{@{}cccccccccccccc@{}}
  \toprule
  \multicolumn{1}{c}{\multirow{2}{*}{Method}} & \multicolumn{1}{c|}{\multirow{2}{*}{{Venue}}}                            & \multicolumn{7}{c}{{BD-Rate (\%) w.r.t. VTM-13.2 LDB~\cite{bross2021overview}}}  \\
  \multicolumn{1}{c}{}      & \multicolumn{1}{c|}{}                                                & \multicolumn{1}{c}{HEVC B}    & \multicolumn{1}{c}{HEVC C}  & \multicolumn{1}{c}{HEVC D} & \multicolumn{1}{c}{HEVC E} & \multicolumn{1}{c}{UVG~\cite{mercat2020uvg}} & \multicolumn{1}{c}{MCL-JCV~\cite{wang2016mcl}} & \multicolumn{1}{c}{Average} \\       \midrule                                      
  \multicolumn{1}{c}{DCVC-TCM~\cite{sheng2022temporal}}    & \multicolumn{1}{c|}{TMM'22}                                       & $+55.4$       & $+97.4$ & $+50.0$       & $+214.2$   & $+60.4$& $+50.7$ & $+88.0$ \\
  \multicolumn{1}{c}{DCVC-HEM~\cite{li2022hybrid}}    & \multicolumn{1}{c|}{ACMMM'22}                                       & $+3.9$       & $+28.4$ & $-1.2$       & $+66.3$   & $+0.5$& $+1.7$ & $+16.6$ \\
  \multicolumn{1}{c}{DCVC-DC~\cite{li2023neural}}        & \multicolumn{1}{c|}{CVPR'23}                                   & $-11.0$       & $+0.2$ &  $-23.9$    & $-7.8$  & $-21.0$  & $-13.0$  & $-12.8$\\
  \multicolumn{1}{c}{DCVC-FM~\cite{li2024neural}}        & \multicolumn{1}{c|}{CVPR'24}                                   & $-11.7$       & $-7.9$ &  $-28.2$    & $-25.8$  & $-23.9$  & $-12.3$  & $-18.3$\\\midrule
  \multicolumn{1}{c}{DCVC-DC$^{*}$ }      & \multicolumn{1}{c|}{Reproduced}  & $-9.2$ & $+6.9$ & $-20.0$ & $+45.1$  & $-12.5$ & $-11.7$   & $-0.2$ \\
  \multicolumn{1}{c}{ECVC }      & \multicolumn{1}{c|}{Ours}   & $\bm{-27.9}$    & $\bm{-18.9}$ &$\bm{-39.0}$  &$\bm{-26.4}$      & $\bm{-38.3}$ & $\bm{-27.7}$ & $\bm{-29.7}$   \\
  \multicolumn{1}{c}{ECVC-FM}      & \multicolumn{1}{c|}{Ours}   & $\bm{-30.9}$     & $\bm{-22.9}$ & $\bm{-40.7}$      & $\bm{-23.4}$ & $\bm{-42.5}$ & $\bm{-33.6}$ & $\bm{-32.3}$ \\\bottomrule  
\end{tabular}}
  \footnotesize               
\caption{{BD-Rate $(\%)$~\cite{bjontegaard2001calculation} comparison for PSNR (dB). The anchor is \textbf{VTM-13.2 LDB}. \textbf{The Intra Period is $\bm{-1}$ with $\bm{96}$ frames.}}}
\label{tab:supp_rd_ip-1} 
\end{table*}
\begin{table*}
  \centering
  \small
  \setlength{\tabcolsep}{3mm}{
  \renewcommand\arraystretch{1}
  \begin{tabular}{@{}cccccccccccccc@{}}
  \toprule
  \multicolumn{1}{c}{\multirow{2}{*}{Method}} & \multicolumn{1}{c|}{\multirow{2}{*}{{Venue}}}                            & \multicolumn{7}{c}{{BD-Rate (\%) w.r.t. VTM-13.2 LDB~\cite{bross2021overview}}}  \\
  \multicolumn{1}{c}{}      & \multicolumn{1}{c|}{}                                                & \multicolumn{1}{c}{HEVC B}    & \multicolumn{1}{c}{HEVC C}  & \multicolumn{1}{c}{HEVC D} & \multicolumn{1}{c}{HEVC E} & \multicolumn{1}{c}{UVG} & \multicolumn{1}{c}{MCL-JCV} & \multicolumn{1}{c}{Average} \\       \midrule                                      
  \multicolumn{1}{c}{DCVC-TCM~\cite{sheng2022temporal}}    & \multicolumn{1}{c|}{TMM'22}                                       & $+107.3$       & $+143.5$ & $+99.2$       & $+835.9$   & $+120.6$& $+63.7$ & $+228.4$ \\
  \multicolumn{1}{c}{DCVC-HEM~\cite{li2022hybrid}}    & \multicolumn{1}{c|}{ACMMM'22}                                       & $+22.8$       & $+32.3$ & $+13.4$       & $+236.9$   & $+33.5$& $+6.7$ & $+57.6$ \\
  \multicolumn{1}{c}{DCVC-DC~\cite{li2023neural}}        & \multicolumn{1}{c|}{CVPR'23}                                   & $-7.5$       & $+3.4$ &  $-12.0$    & $+83.9$  & $-4.5$  & $-12.9$  & $+8.4$\\
  \multicolumn{1}{c}{DCVC-FM~\cite{li2024neural}}        & \multicolumn{1}{c|}{CVPR'24}                                   & $-19.9$       & $-17.4$ &  $-25.7$    & $-24.5$  & $-22.5$  & $-13.4$  & $-20.6$\\\midrule
  \multicolumn{1}{c}{DCVC-DC$^{*}$ }      & \multicolumn{1}{c|}{Reproduced}  & $+11.9$ & $+11.6$ & $-3.6$ & $+446.4$  & $+16.2$ & $-9.8$   & $+78.8$ \\
  \multicolumn{1}{c}{ECVC }      & \multicolumn{1}{c|}{Ours}   & $\bm{-33.4}$    & $\bm{-29.5}$ &$\bm{-38.8}$  &$\bm{-23.5}$      & $\bm{-37.5}$ & $\bm{-29.7}$ & $\bm{-32.1}$   \\
  \multicolumn{1}{c}{ECVC-FM}      & \multicolumn{1}{c|}{Ours}   & $\bm{-33.5}$     & $\bm{-32.7}$ & $\bm{-41.3}$      & $\bm{-17.3}$ & $\bm{-41.3}$ & $\bm{-34.3}$ & $\bm{-33.4}$\\\bottomrule  
\end{tabular}}
  \footnotesize               
\caption{{BD-Rate $(\%)$~\cite{bjontegaard2001calculation} comparison for PSNR (dB). The anchor is \textbf{VTM-13.2 LDB}. \textbf{The Intra Period is $\bm{-1}$ with All frames.}}}
\label{tab:supp_rd_ip-1_all} 
\end{table*}
\section{Rate-Distortion Results}
We provide the peroformances of our baseline model, reproduced DCVC-DC$^{\ast}$ in Table~\ref{tab:supp_rd}, Table~\ref{tab:supp_rd_ip32_ssim}, Table~\ref{tab:supp_rd_ip-1} and Table~\ref{tab:supp_rd_ip-1_all}.
The bpp-PSNR curves are presented in Figure~\ref{fig:rd_ip32}, Figure~\ref{fig:rd_ip-1} and Figure~\ref{fig:rd_ip-1_all}.
The bpp-MS-SSIM curves are presented in Figure~\ref{fig:rd_ip32_ssim}.
\begin{figure*}[ht]
    \centering
    \subfloat{
      \includegraphics[scale=0.45]{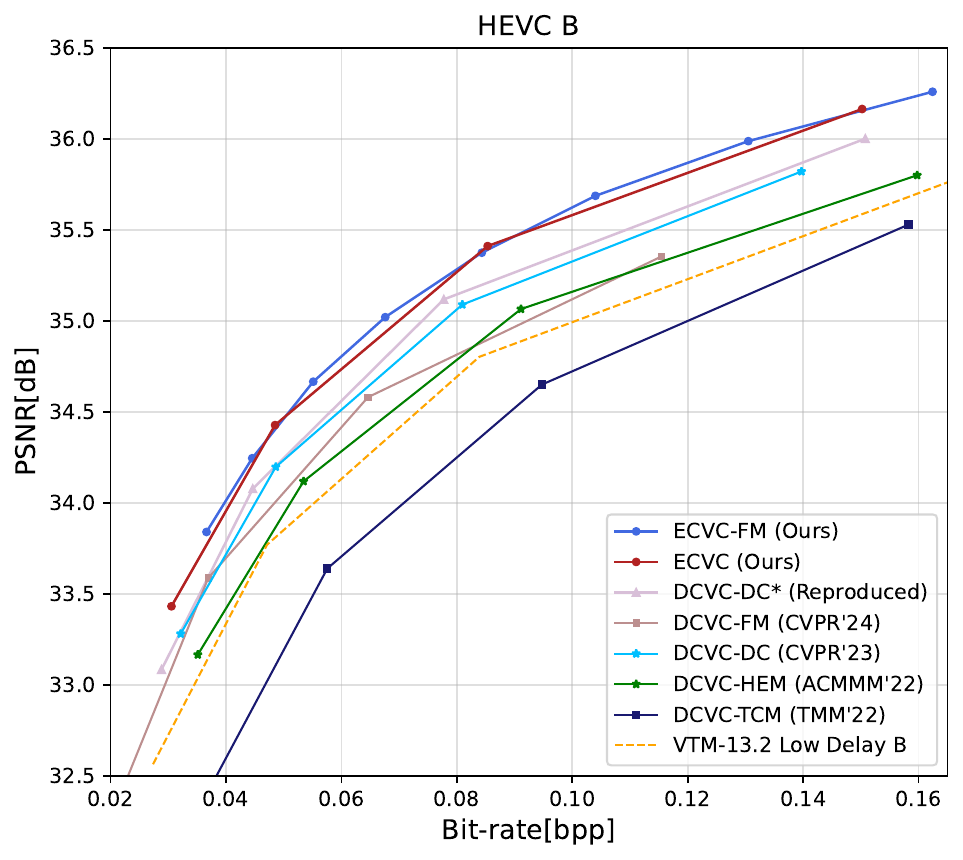}}
    \subfloat{
    \includegraphics[scale=0.45]{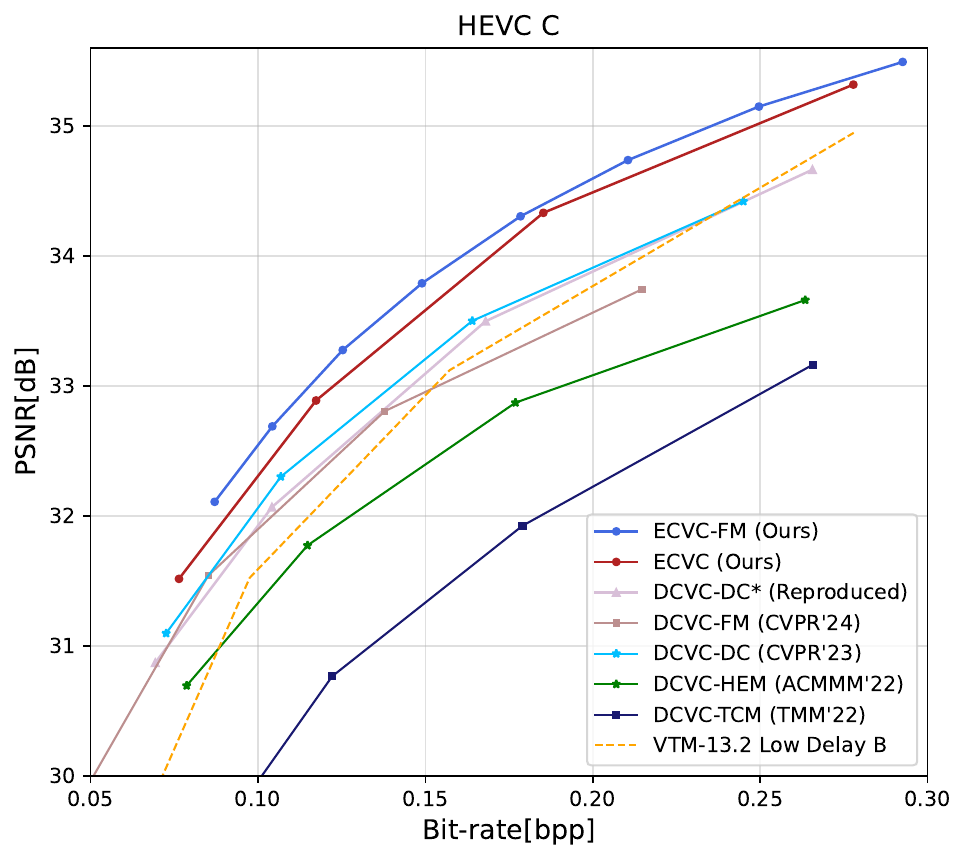}}\\
    \subfloat{
    \includegraphics[scale=0.45]{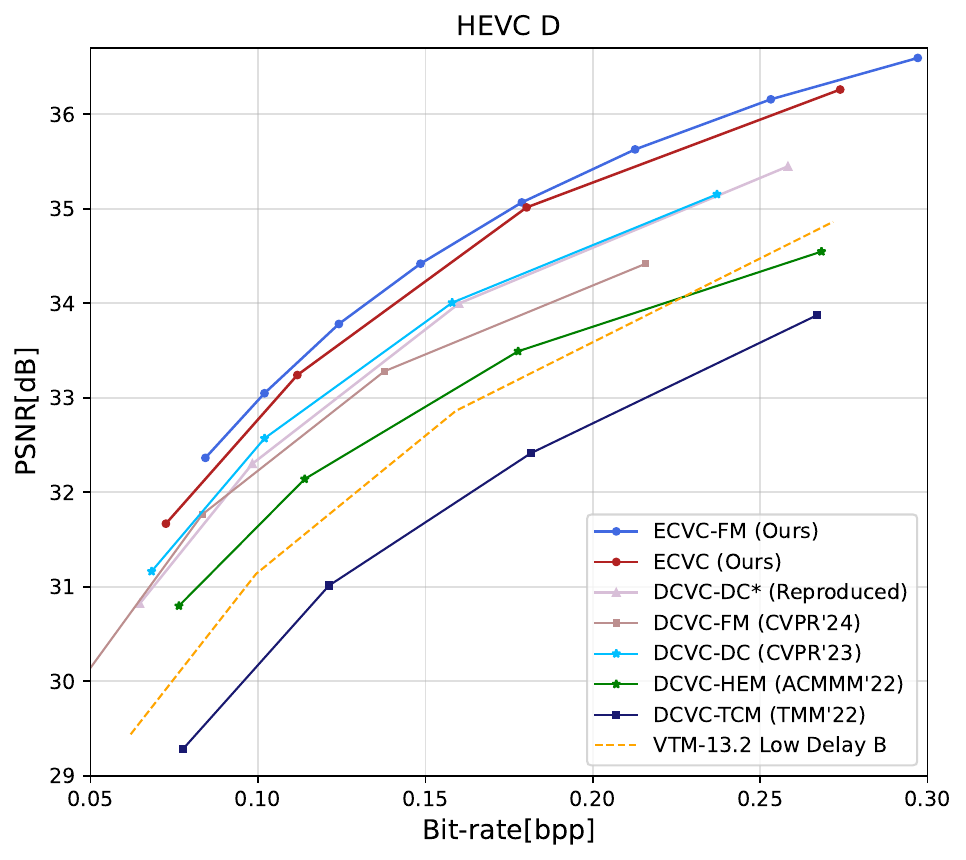}}
    \subfloat{
      \includegraphics[scale=0.45]{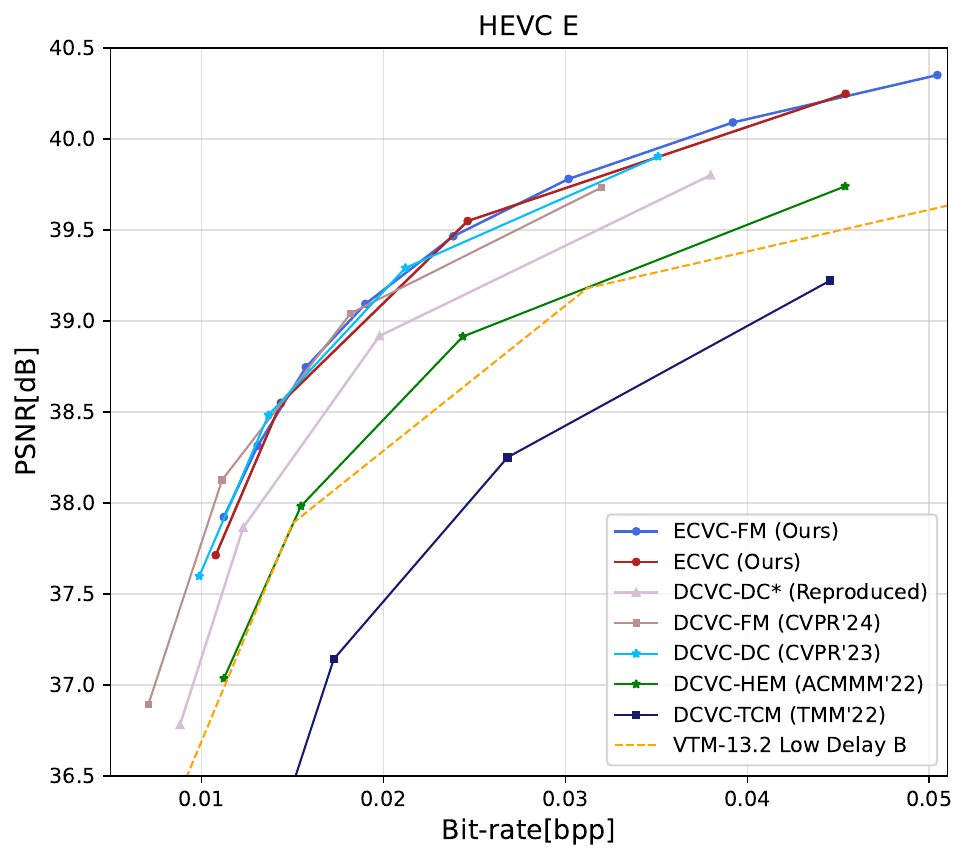}}\\
      \subfloat{
        \includegraphics[scale=0.45]{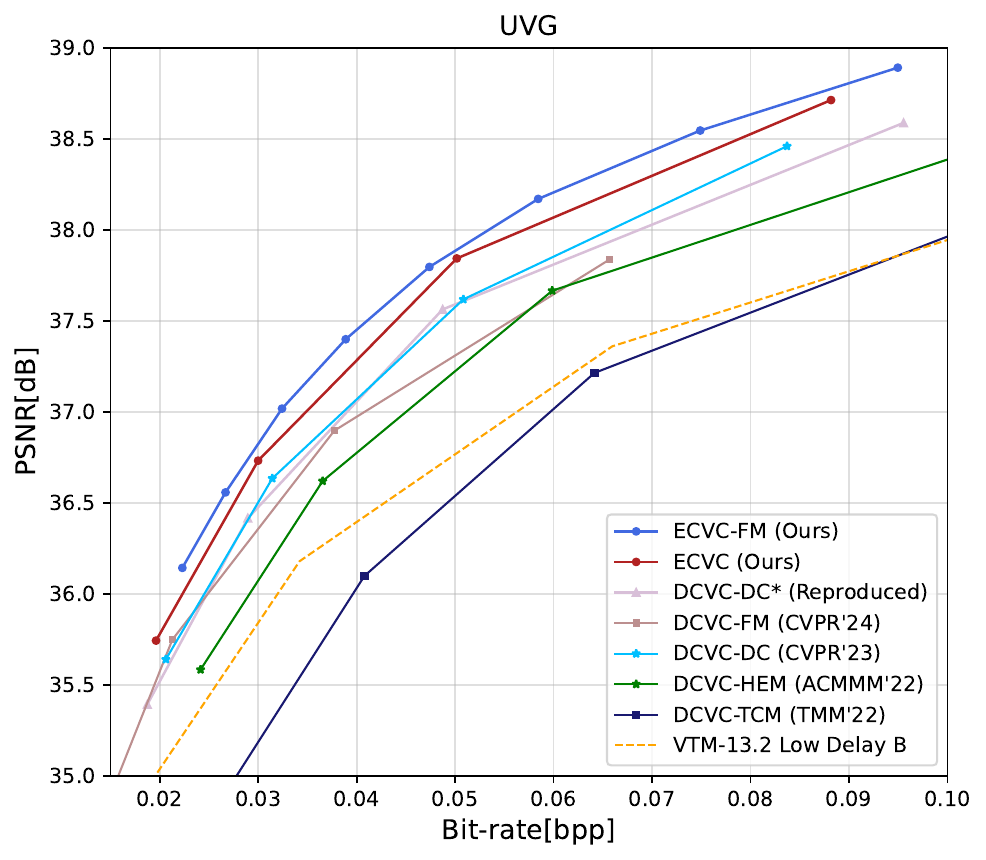}}
        \subfloat{
          \includegraphics[scale=0.45]{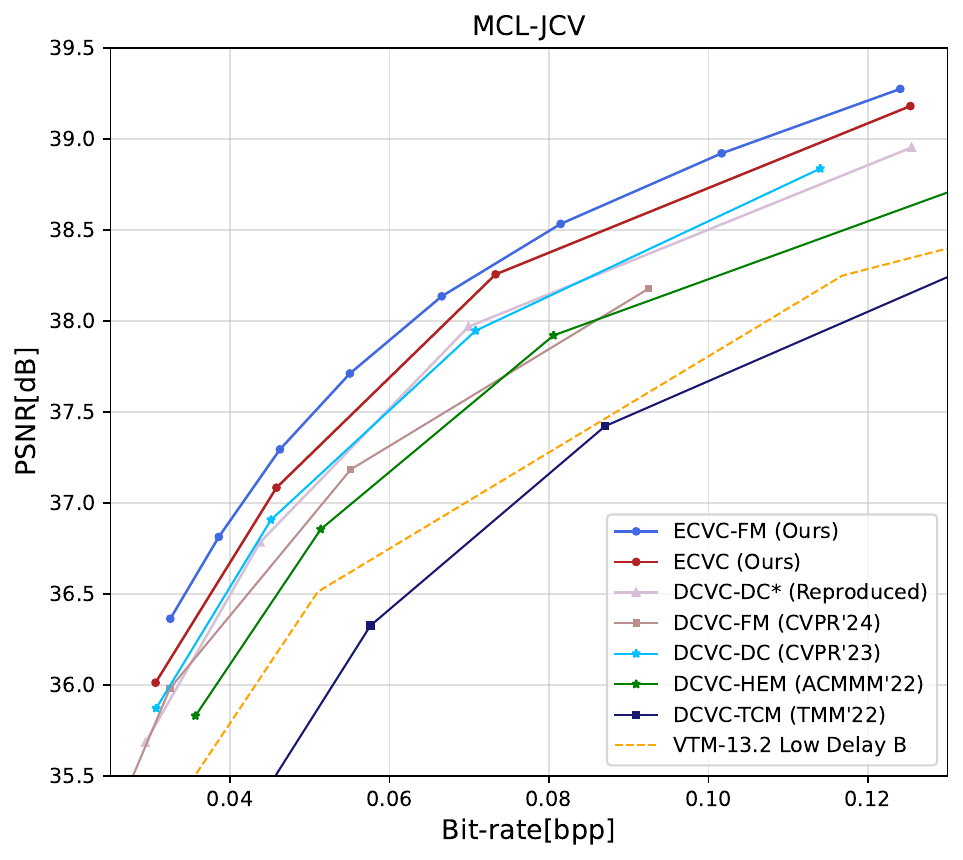}}\\
    \caption{Bpp-PSNR curves on HEVC B, C, D, E, UVG and MCL-JCV dataset. \textbf{The intra period is $\bm{32}$ with $\bm{96}$ frames}.}
    \label{fig:rd_ip32}
  \end{figure*}
  \begin{figure*}[ht]
    \centering
    \subfloat{
      \includegraphics[scale=0.45]{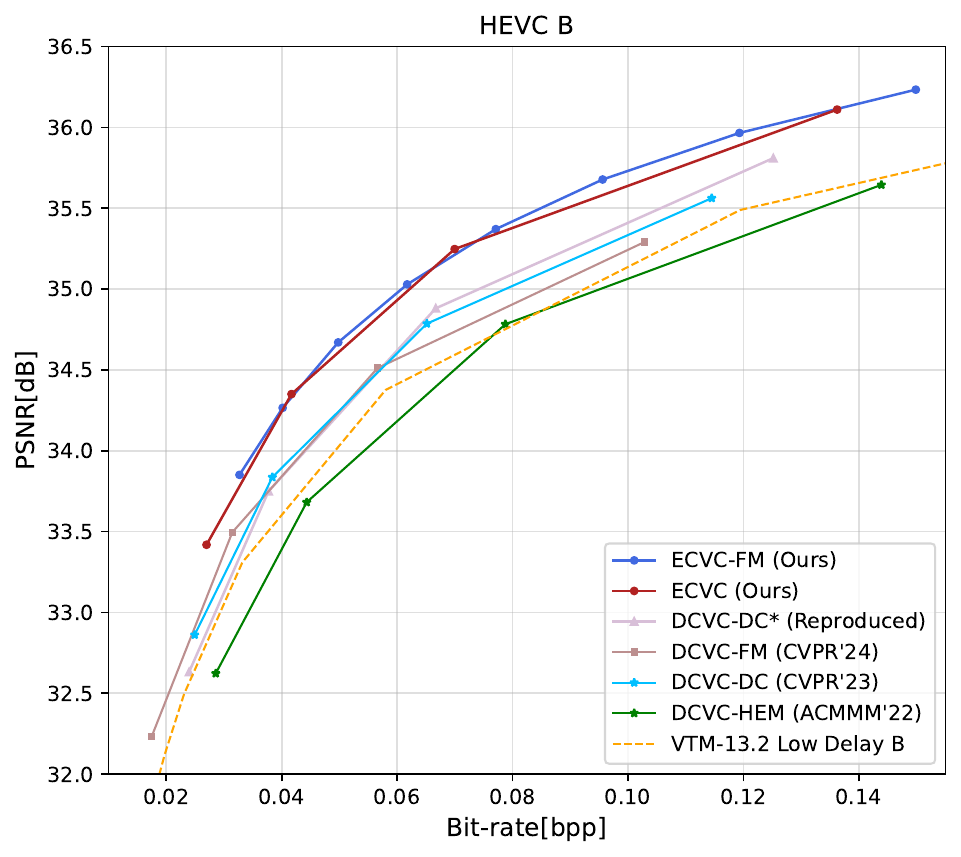}}
    \subfloat{
    \includegraphics[scale=0.45]{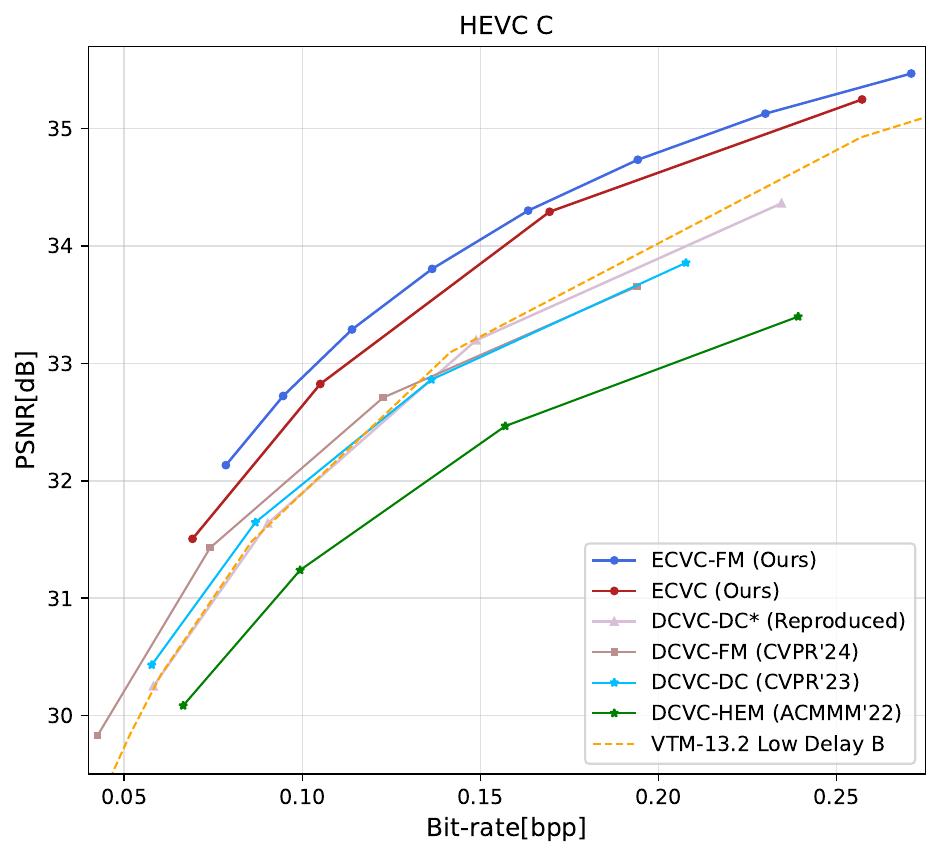}}\\
    \subfloat{
    \includegraphics[scale=0.45]{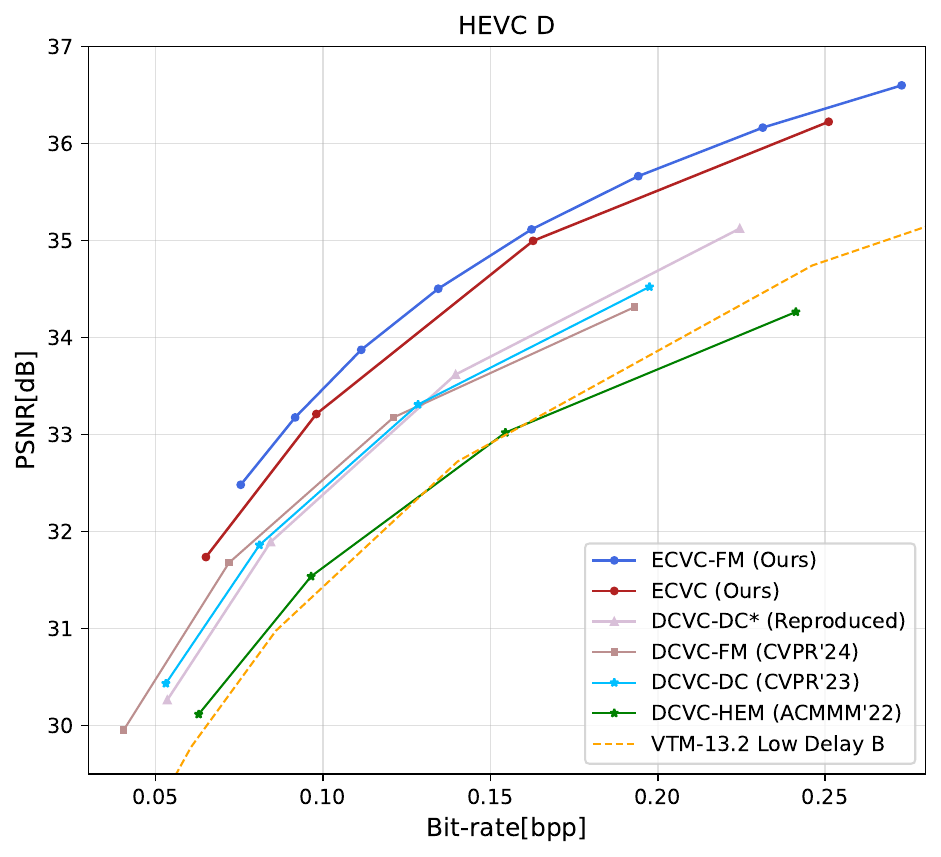}}
    \subfloat{
      \includegraphics[scale=0.45]{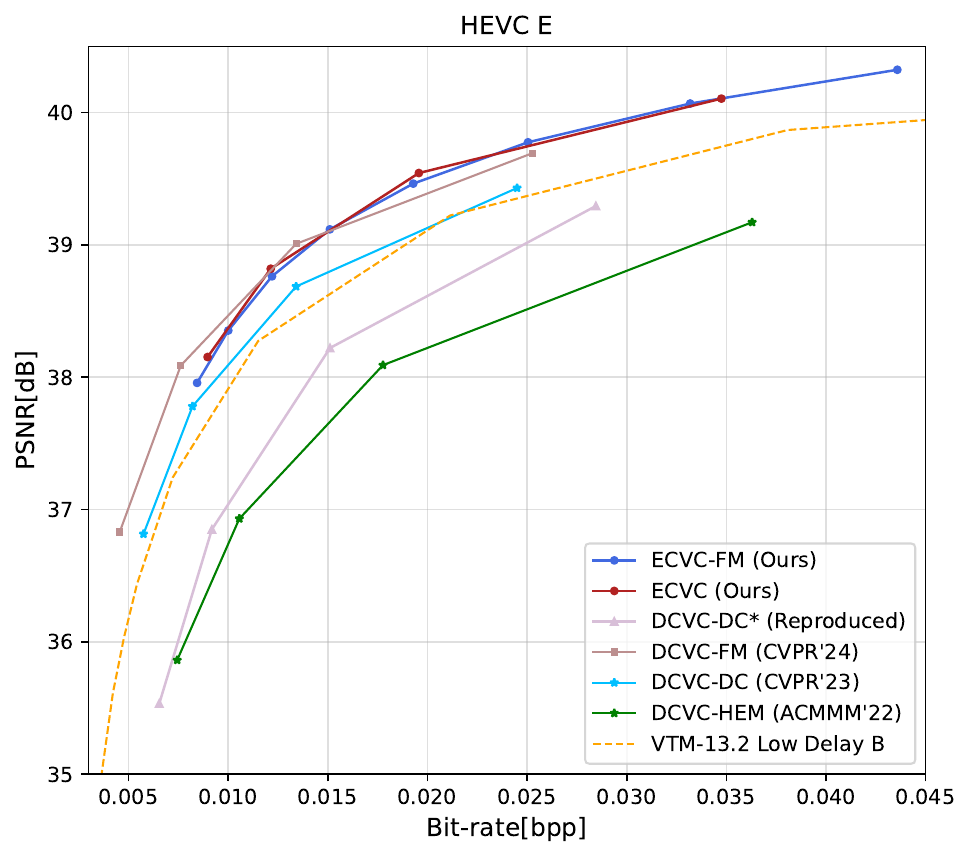}}\\
      \subfloat{
        \includegraphics[scale=0.45]{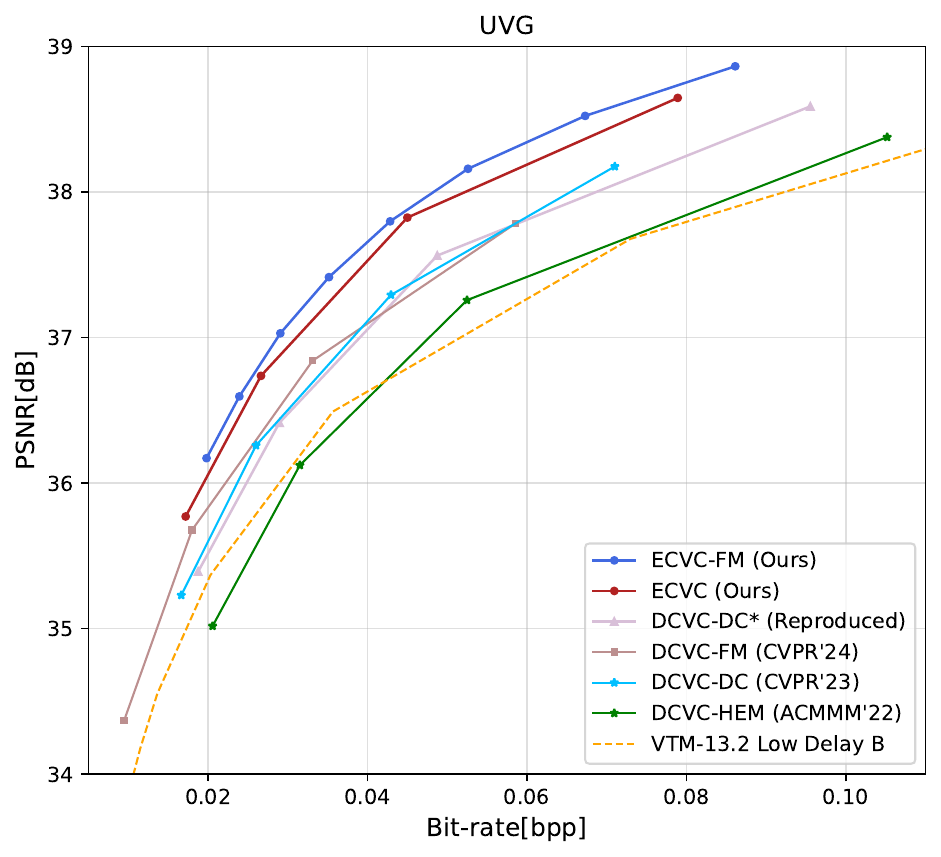}}
        \subfloat{
          \includegraphics[scale=0.45]{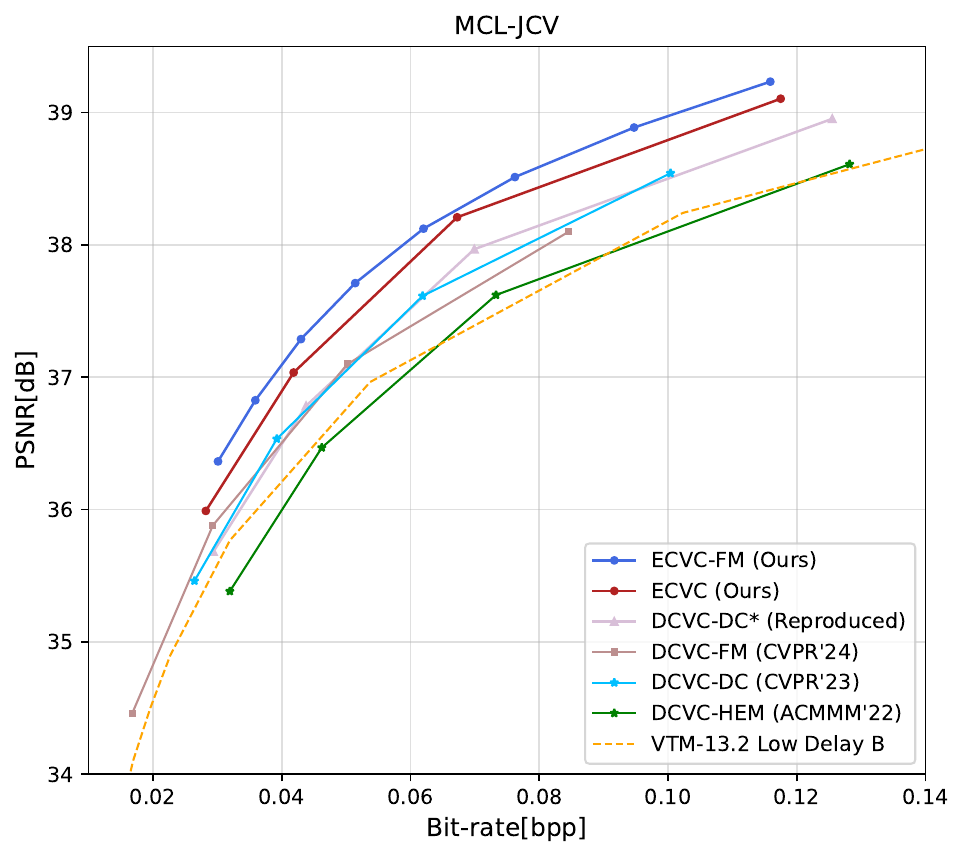}}\\
    \caption{Bpp-PSNR curves on HEVC B, C, D, E, UVG and MCL-JCV dataset. \textbf{The intra period is $\bm{-1}$ with $\bm{96}$ frames}.}
    \label{fig:rd_ip-1}
  \end{figure*}
  \begin{figure*}[ht]
    \centering
    \subfloat{
      \includegraphics[scale=0.45]{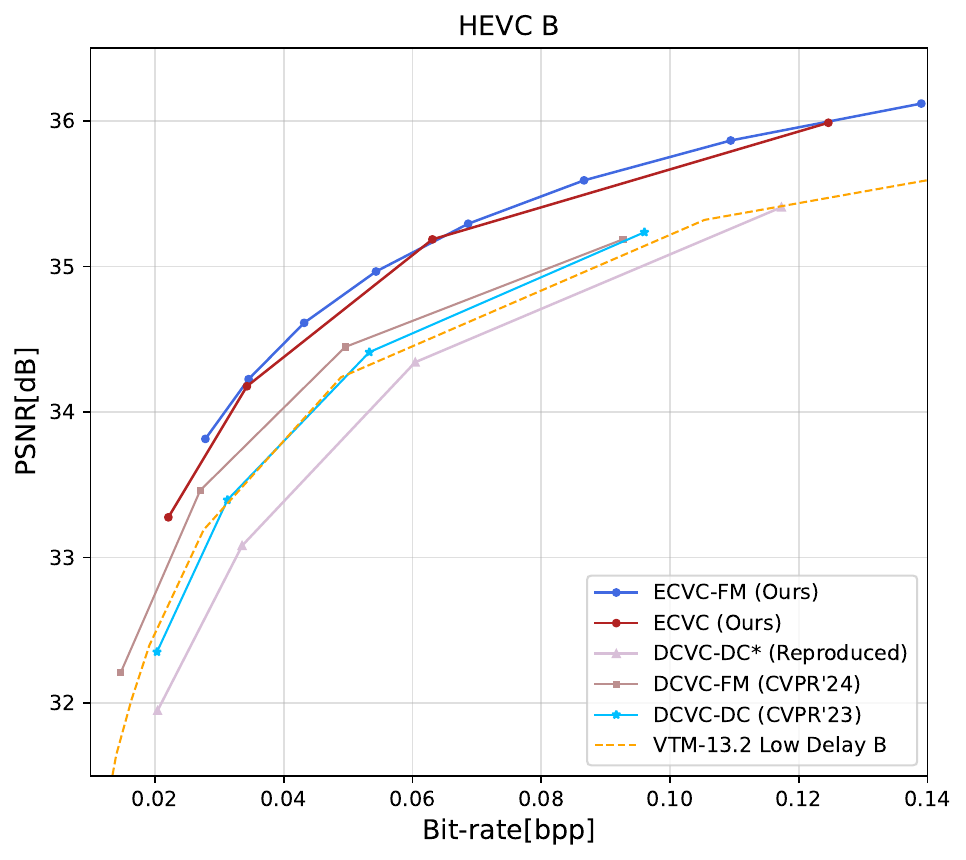}}
    \subfloat{
    \includegraphics[scale=0.45]{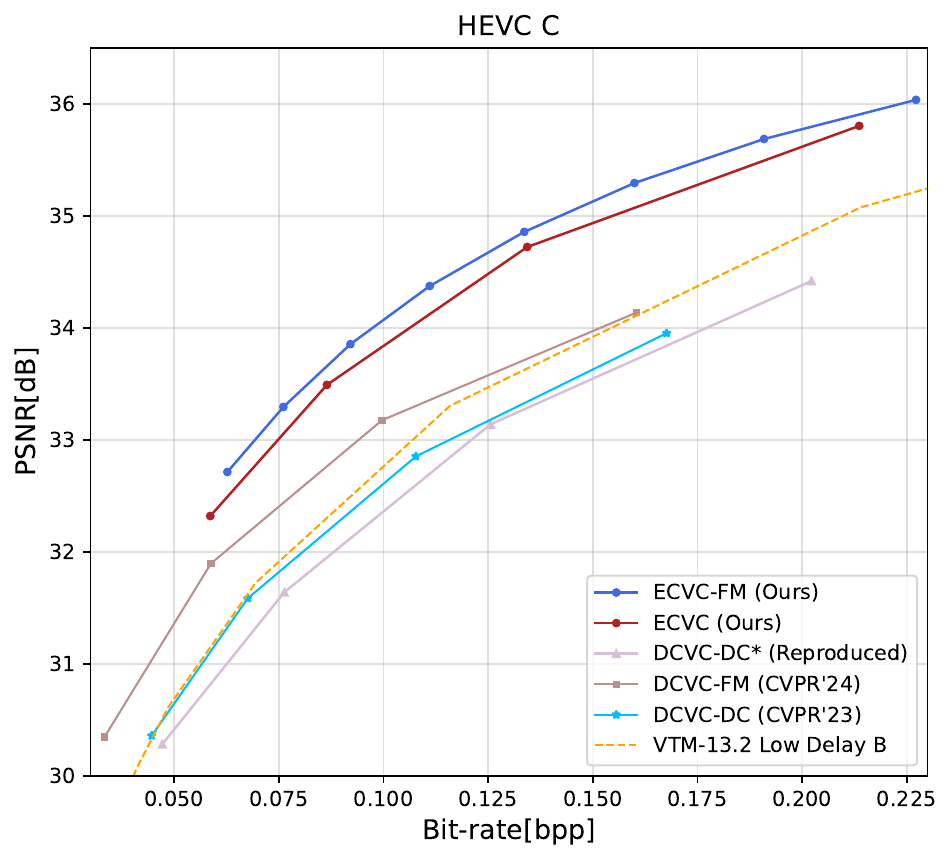}}\\
    \subfloat{
    \includegraphics[scale=0.45]{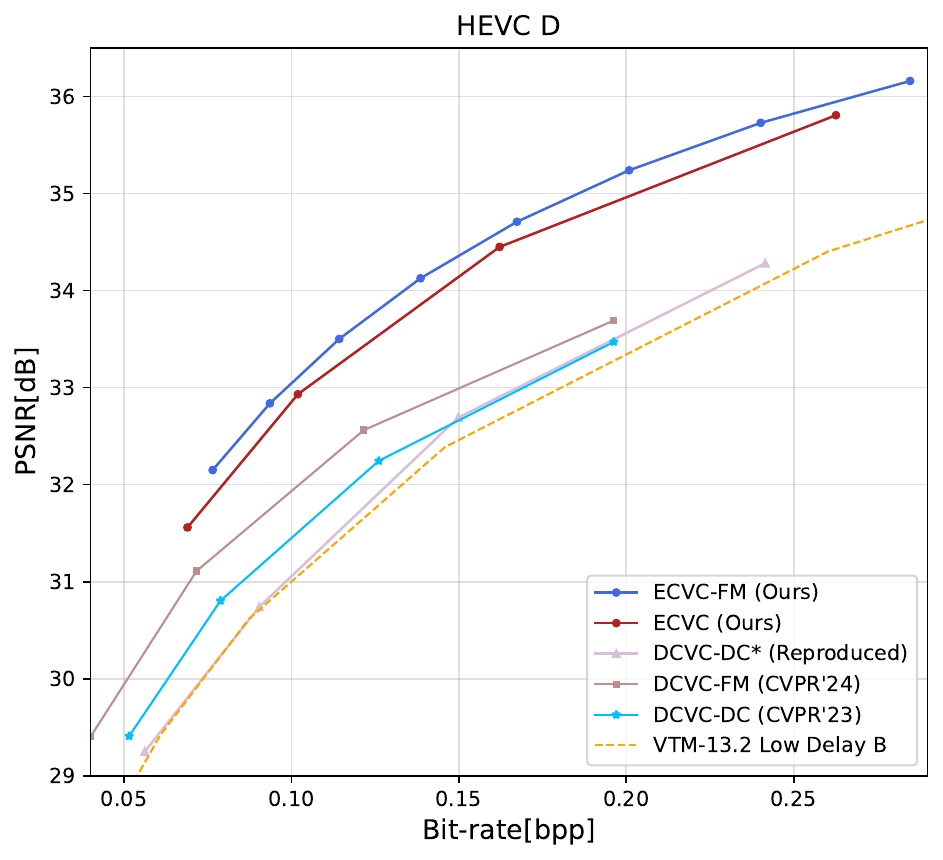}}
    \subfloat{
      \includegraphics[scale=0.45]{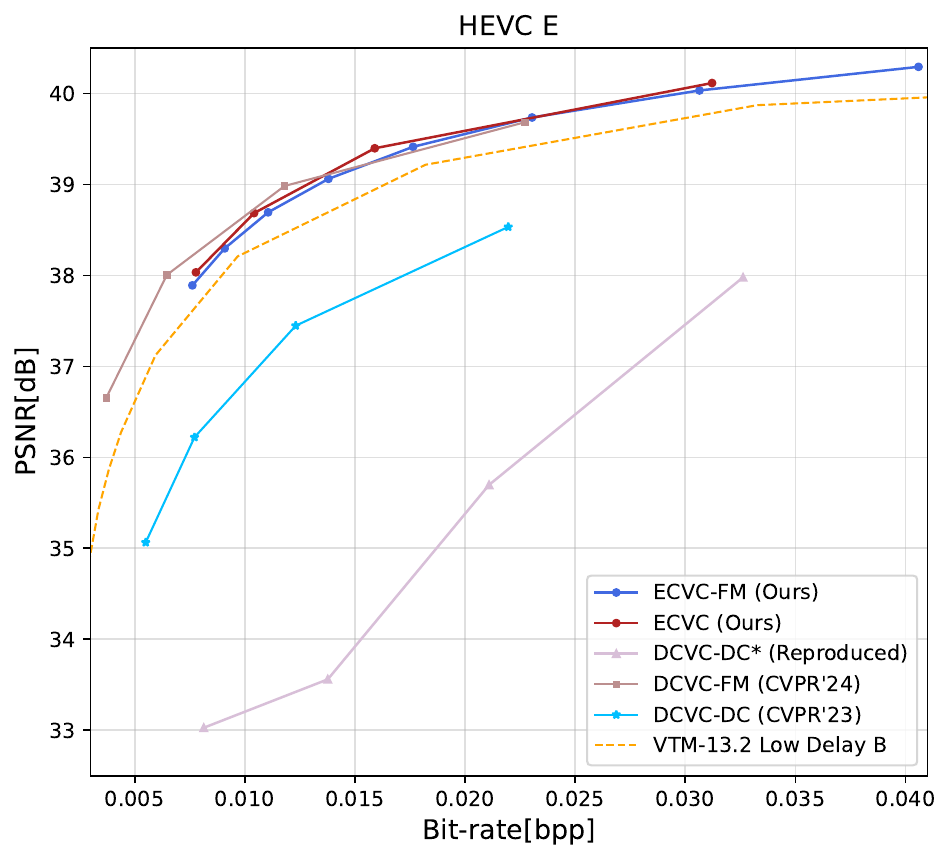}}\\
      \subfloat{
        \includegraphics[scale=0.45]{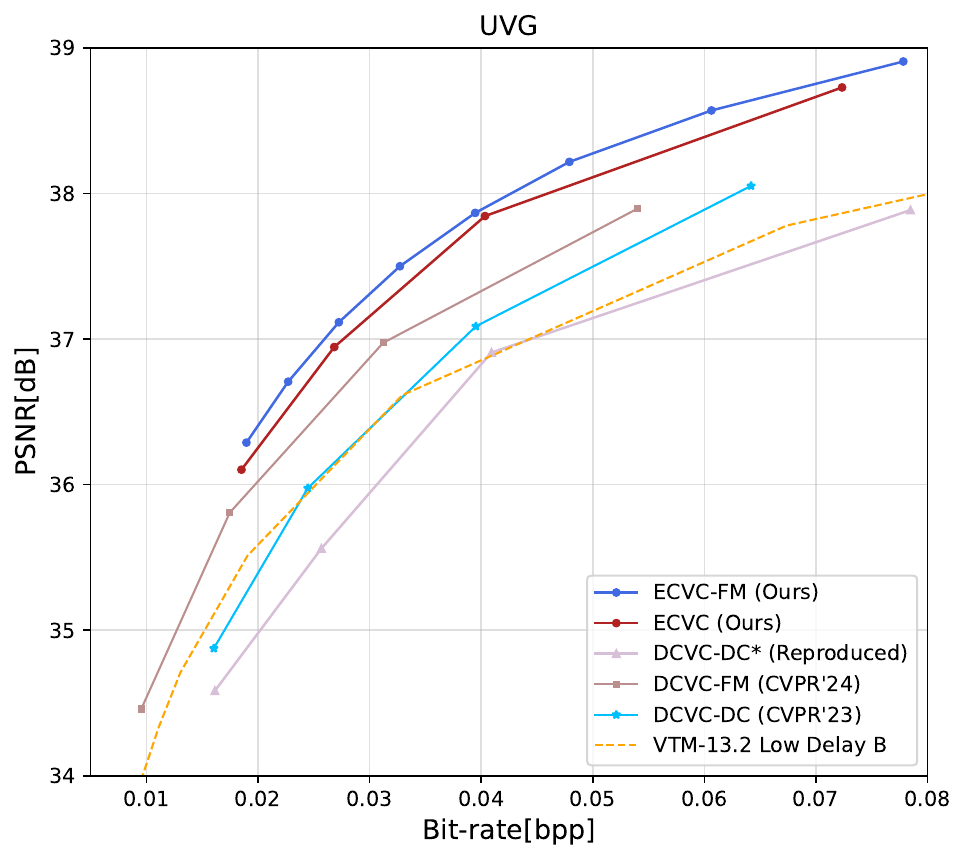}}
        \subfloat{
          \includegraphics[scale=0.45]{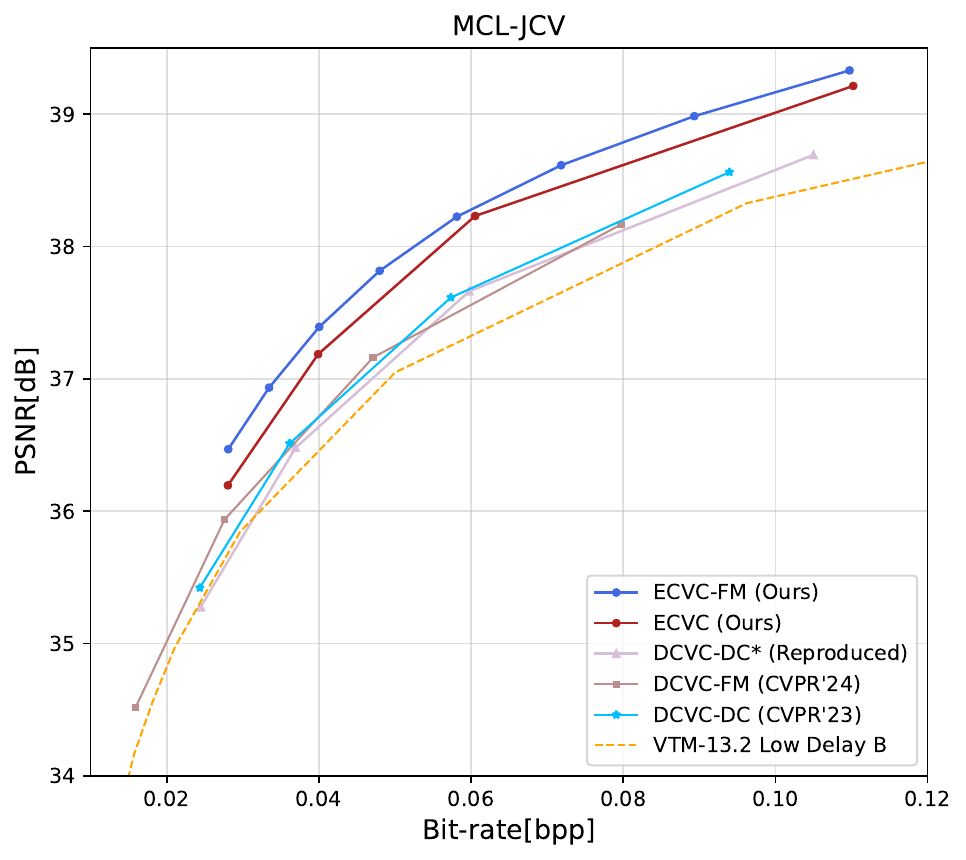}}\\
    \caption{Bpp-PSNR curves on HEVC B, C, D, E, UVG and MCL-JCV dataset. \textbf{The intra period is $\bm{-1}$ with All frames}.}
    \label{fig:rd_ip-1_all}
  \end{figure*}
  \begin{figure*}[ht]
    \centering
    \subfloat{
      \includegraphics[scale=0.45]{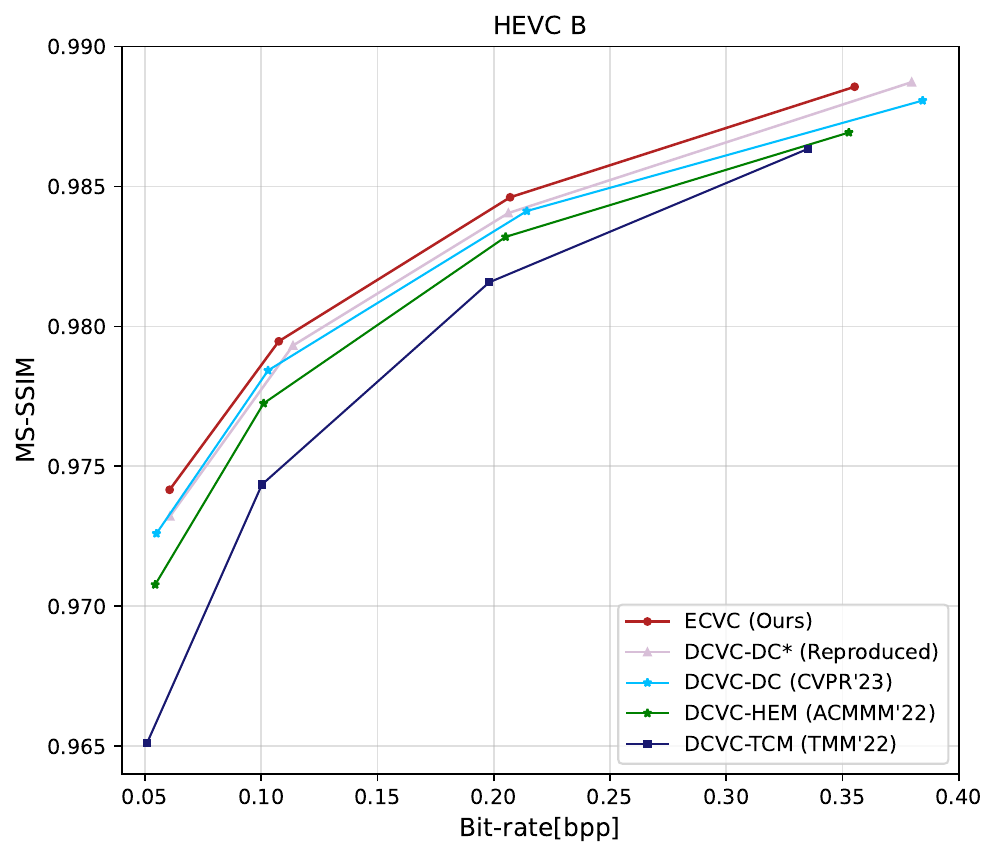}}
    \subfloat{
    \includegraphics[scale=0.45]{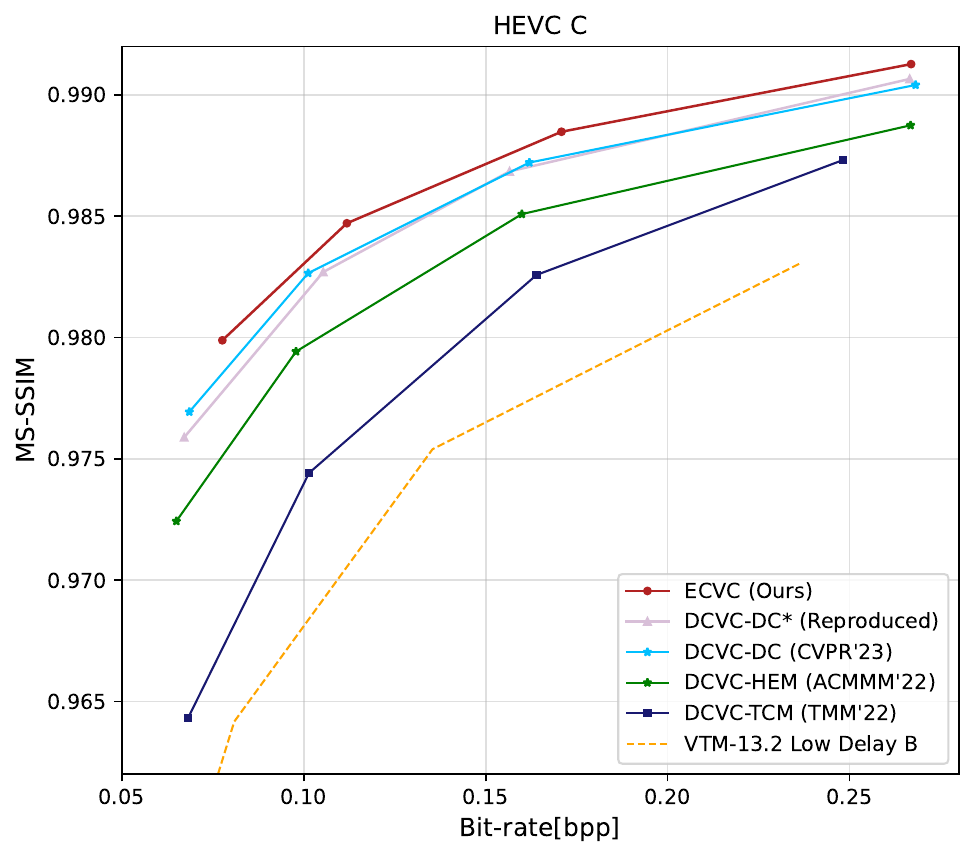}}\\
    \subfloat{
    \includegraphics[scale=0.45]{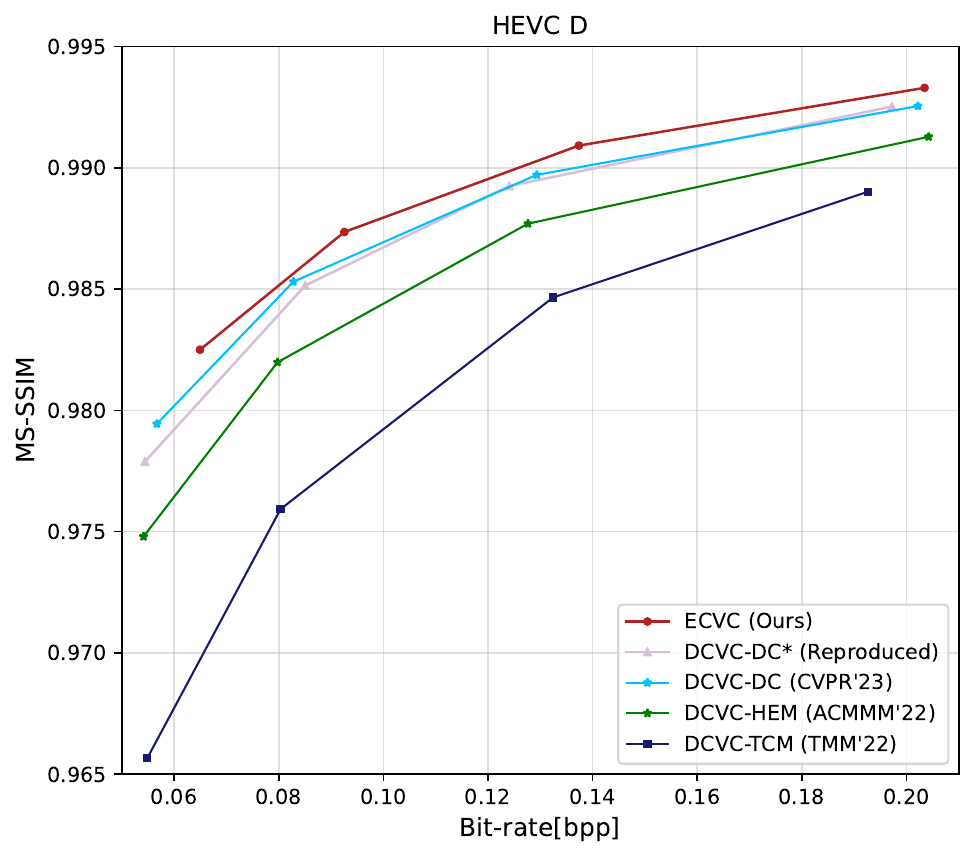}}
    \subfloat{
      \includegraphics[scale=0.45]{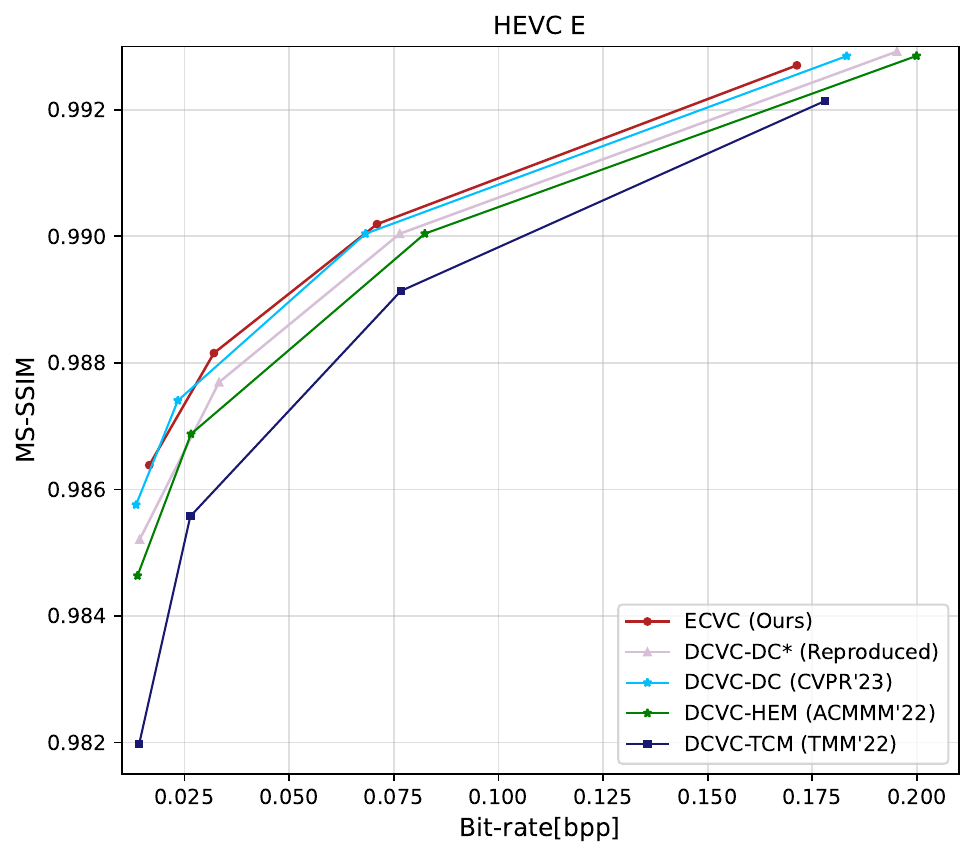}}\\
      \subfloat{
        \includegraphics[scale=0.45]{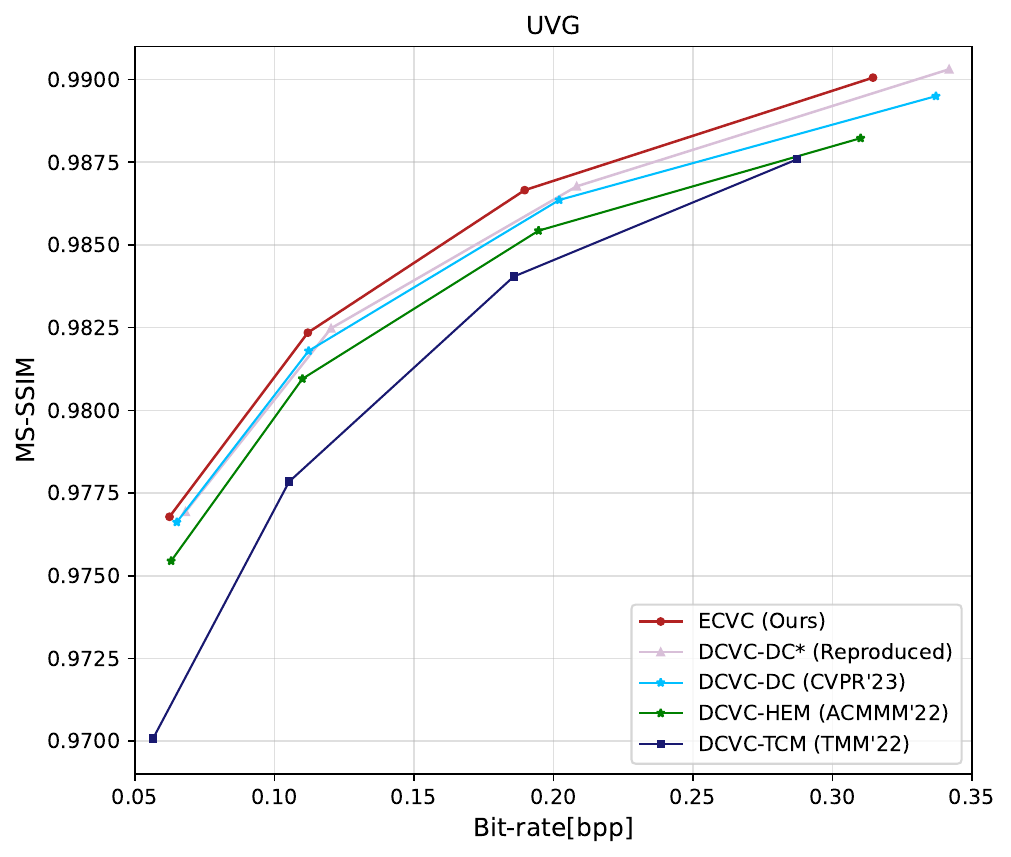}}
        \subfloat{
          \includegraphics[scale=0.45]{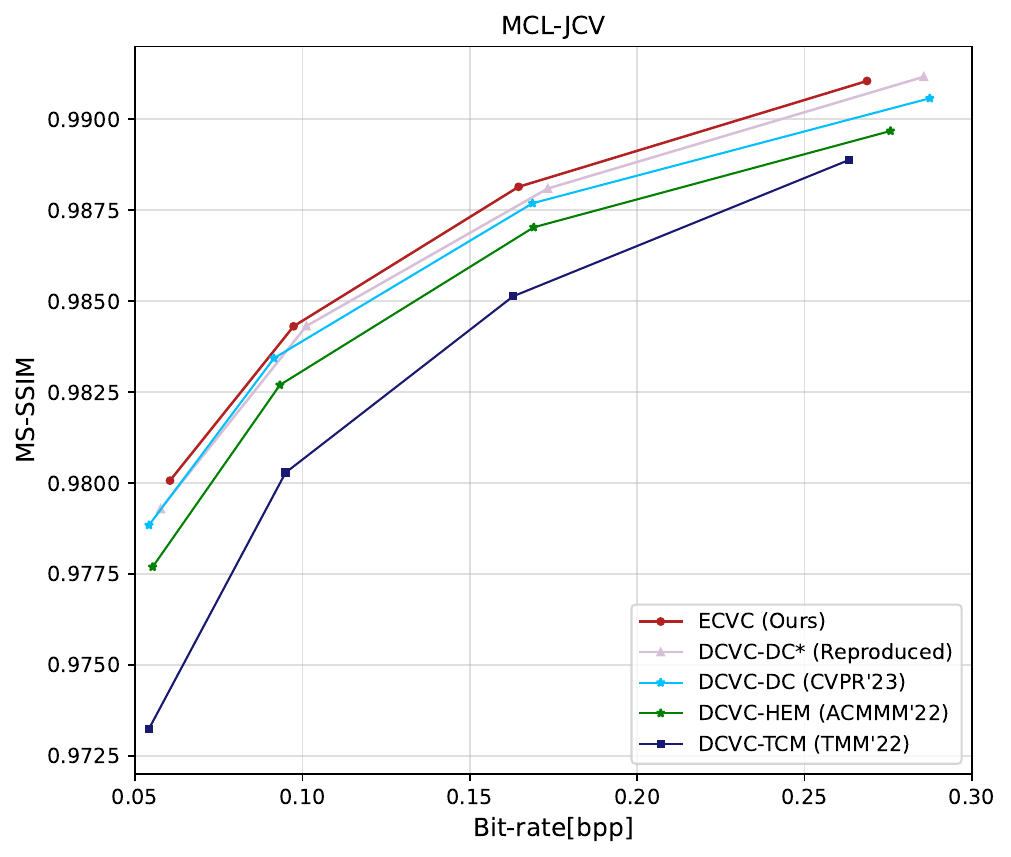}}\\
    \caption{Bpp-MS-SSIM curves on HEVC B, C, D, E, UVG and MCL-JCV dataset. \textbf{The intra period is $\bm{32}$ with 96 frames}.}
    \label{fig:rd_ip32_ssim}
  \end{figure*}

\end{document}